\documentclass[a4paper,11pt]{article}
\usepackage[utf8]{inputenc}
\usepackage[T1]{fontenc}
\usepackage{libertine} 
\usepackage{inconsolata} 
\usepackage{needspace}
\usepackage{fullpage}
\usepackage{amssymb,amsfonts}%
\usepackage[tbtags]{amsmath}%
\usepackage{amsthm}%
\usepackage{hyperref}%
\usepackage[capitalise]{cleveref}%
\crefname{subsection}{Subsection}{Subsections}
\crefname{item}{Item}{Items}

\usepackage[table]{xcolor}

\usepackage{graphicx}%
\usepackage{multirow}%
\usepackage{url}%
\usepackage{mathrsfs}%
\usepackage[title]{appendix}%
\usepackage{textcomp}%
\usepackage{manyfoot}%
\usepackage{booktabs}%
\usepackage{algorithm}%
\usepackage{algorithmicx}%
\usepackage{algpseudocode}%
\usepackage{listings}%
\usepackage{comment}%
\usepackage{multirow}%
\usepackage{authblk}%
\usepackage{thmtools}%
\usepackage{orcidlink}%
\usepackage{mathtools}%
\usepackage{tabularx,lipsum,environ}%
\usetikzlibrary{arrows.meta,matrix}%
\usepackage{nicematrix}%
\usepackage{etoolbox,refcount}
\usepackage{multicol}
\usepackage{tikz}%
\usepackage[size=footnotesize]{todonotes}
\usepackage{microtype} 

\newcounter{countitems}
\newcounter{nextitemizecount}
\newcommand{\setupcountitems}{%
  \stepcounter{nextitemizecount}%
  \setcounter{countitems}{0}%
  \preto\item{\stepcounter{countitems}}%
}
\makeatletter
\newcommand{\computecountitems}{%
  \edef\@currentlabel{\number\c@countitems}%
  \label{countitems@\number\numexpr\value{nextitemizecount}-1\relax}%
}
\newcommand{\nextitemizecount}{%
  \getrefnumber{countitems@\number\c@nextitemizecount}%
}
\newcommand{\previtemizecount}{%
  \getrefnumber{countitems@\number\numexpr\value{nextitemizecount}-1\relax}%
}
\makeatother    
\newenvironment{AutoMultiColItemize}{%
\ifnumcomp{\nextitemizecount}{>}{3}{\begin{multicols}{2}}{}%
\setupcountitems\begin{itemize}}%
{\end{itemize}%
\unskip\computecountitems\ifnumcomp{\previtemizecount}{>}{3}{\end{multicols}}{}}

\long\def\mbox#1{\leavevmode\hbox{#1}}

\makeatletter
\newcommand{\problemtitle}[1]{\gdef\@problemtitle{#1}}
\newcommand{\probleminput}[1]{\gdef\@probleminput{#1}}
\newcommand{\problemoutput}[1]{\gdef\@problemoutput{#1}}
\NewEnviron{problem}{
  \problemtitle{}\probleminput{}\problemoutput{}
  \BODY
  \par\addvspace{.5\baselineskip}
  \noindent
  \begin{tabularx}{\textwidth}{@{\hspace{\parindent}} l X c}
    \multicolumn{2}{@{\hspace{\parindent}}l}{\@problemtitle} \\
    \textbf{Input:} & \@probleminput \\
    \textbf{Output:} & \@problemoutput
  \end{tabularx}
  \par\addvspace{.5\baselineskip}
}
\makeatother

\newcommand\restr[2]{{
  \left.\kern-\nulldelimiterspace 
  #1 
  \littletaller 
  \right|_{#2} 
  }}

\newcommand{\littletaller}{\mathchoice{\vphantom{\big|}}{}{}{}}
\newcommand{\Z}{\mathbb{Z}}
\newcommand{\PDS}{PDS}
\newcommand{\maxPDS}{\textsc{MaxPDS}}
\newcommand{\connmaxPDS}{\textsc{connected MaxPDS}}

\newcommand{\NP}{\textsf{NP}}

\newcommand{\degen}{\operatorname{degen}}

\newtheorem{theorem}{Theorem}
\newtheorem{claim}{Claim}%
\newtheorem{remark}{Remark}%
\newtheorem{definition}{Definition}%
\newtheorem{corollary}{Corollary}%
\newtheorem{lemma}{Lemma}%

\newcommand{\intInterval}[1]{[\![#1]\!]}
\title{Proportionally dense subgraphs of maximum size in degree-constrained graphs}

\usepackage{authblk}

\author[1]{Narmina Baghirova}
\author[2]{Antoine Castillon}
\affil[1]{Department of Informatics, University of Fribourg, Switzerland}
\affil[2]{Univ.\ Lille, CNRS, Centrale Lille, UMR 9189 CRIStAL, F-59000 Lille, France}

\date{}

\begin{document}
\definecolor{vert}{HTML}{008800}
\maketitle
\begin{abstract}
A \emph{proportionally dense subgraph (\PDS{})} of a graph is an induced subgraph of size at least two such that every vertex in the subgraph has proportionally as many neighbors inside as outside of the subgraph. Then, \maxPDS{} is the problem of determining a \PDS{} of maximum size in a given graph. If we further require that a \PDS{} induces a connected subgraph, we refer to such problem as \connmaxPDS{}. In this paper, we study the complexity of \maxPDS{} with respect to parameters representing the density of a graph and its complement. We consider $\Delta$, representing the maximum degree, $h$, representing the $h$-index, and $\degen{}$, representing the degeneracy of a graph. We show that \maxPDS{} is \NP{}-hard parameterized by $\Delta,h$ and $\degen{}$. More specifically, we show that \maxPDS{} is \NP{}-hard on graphs with $\Delta=4$, $h=4$ and $\degen =2$. Then, we show that \maxPDS{} is \NP{}-hard when restricted to dense graphs, more specifically graphs $G$ such that $\Delta(\overline{G})\leq 6$, and graphs $G$ such that $\degen(\overline{G}) \leq 2$ and $\overline{G}$ is bipartite, where $\overline{G}$ represents the complement of $G$.
On the other hand, we show that \maxPDS{} is polynomial-time solvable on graphs with $h\le2$. Finally, we consider graphs $G$ such that $h(\overline{G})\le 2$ and show that there exists a polynomial-time algorithm for finding a \PDS{} of maximum size in such graphs. This result implies polynomial-time complexity on graphs with $n$ vertices of minimum degree $n-3$, i.e.\ graphs $G$ such that $\Delta(\overline{G})\le 2$. For each result presented in this paper, we consider \connmaxPDS{} and explain how to extend it when we require connectivity.
\end{abstract}


\section{Introduction}
\label{sec:introduction}
\subsection{Motivation and Known Results}
\label{subsec:motivationandknownresults}
In this paper, we study the so-called, \textsc{Maximum Proportionally Dense Subgraph (MaxPDS)} problem, introduced in \cite{MR4023158} and motivated by community detection in networks. 
Community detection has diverse applications, ranging from developing social media algorithms for identifying individuals with shared interests to discovering a set of functionally associated proteins in protein–protein interaction networks in the field of bioinformatics (for instance, see \cite{CommunityProtein}). 
By representing networks as graphs, with vertices representing individuals/elements and edges representing connections between them, we can analyze the network from a structural perspective. In 2013, Olsen (see~\cite{olsenGenView}) argued that the notion of \emph{proportionality} in the definition of community is both intuitive and supported by observations in real-world networks.  
The definition studied in this paper captures the proportion of neighbors for each vertex within the subgraph and, most importantly, compares it with the proportion of neighbors outside of the subgraph. In this way, important properties of a community are captured. Below, we give definitions of problems of a similar flavour in the literature that fail to capture \textit{natural} properties of community.

Two examples of such problems are, for instance, \textsc{Densest k-subgraph} and \textsc{$\gamma$-Quasi-Cliques}. \textsc{Densest k-Subgraph} (see~\cite{densestksub}) is the problem of finding a subgraph of maximum density containing exactly $k$ vertices, where the density of a subgraph $S\subseteq V$ is commonly defined as $\frac{\vert E(S)\vert }{\vert S \vert}$, where $E(S)$ is the set of edges of $G$ with both endpoints in $S$. 
Consider the example in which a vertex is part of the densest k-subgraph according to the definition, however, the vertex is adjacent to $10$ out of $20$ vertices in the subgraph but to $20$ out of $25$ vertices outside of the subgraph. Intuitively, considering the proportionality, it seems natural that, since the vertex is adjacent to a smaller proportion of vertices in the subgraph than outside, it should not be a part of the subgraph. 
Then, \textsc{$\gamma$-Quasi-Cliques}, for some given constant $\gamma\in\, ]0,1[$, is the problem of finding an induced subgraph of maximum size with edge density at least $\gamma$. Even when $\gamma$ is close to $1$, the example used earlier applies similarly to the $\gamma$-Quasi-Cliques. There are other related problems in the literature, see for instance \cite{ListOfSimilarProblems}, for which we can similarly deduce that they may not necessarily capture the intuitive characteristics of the notion of community. 


In this paper, we study the \textsc{Maximum Proportionally Dense Subgraph (MaxPDS)} problem, introduced in \cite{MR4023158}, which captures the important properties of a community that the definitions of the previously mentioned problems fail to capture. More specifically, it captures the proportion of neighbors for each vertex within the subgraph and, most importantly, compares it with the proportion of neighbors outside of the subgraph. Let us formally define this problem.

\begin{definition}(Bazgan et al.~\cite{MR4023158})\\
    Let $G=(V,E)$ be a graph and $S\subset V$, such that $2 \le \vert S \vert < \vert V \vert$. The induced subgraph $G[S]$ is a \emph{proportionally dense subgraph (PDS)} if for each $v\in S$:
    $$\frac{d_S(u)}{\vert S \vert -1} \ge \frac{d_{\overline{S}}(u)}{\vert \overline{S} \vert}\,,$$
    which is equivalent to $\displaystyle\frac{d_S(u)}{\vert S \vert -1} \ge \frac{d(u)}{\vert V \vert -1}$ and to $\displaystyle\frac{d(u)}{\vert V \vert -1} \ge \frac{d_{\overline{S}}(u)}{\vert \overline{S} \vert}$, where $\overline{S} = V \setminus S$.
\end{definition}

Notice that the proof of the first equivalence can be found in \cite{bazgan:k-comm}. The second equivalence follows from $\displaystyle\frac{d_S(u)}{\vert S \vert -1} \ge \frac{d(u)}{\vert V \vert -1}$ and the fact that $d_S(u)= d(u) - d_{\overline{S}}(u)$ and $\vert S\vert = |V|- |\overline{S}|$.

In other words, a \PDS{} is an induced subgraph of size at least two such that every vertex of the subgraph is adjacent to proportionally at least as many vertices in the subgraph as to vertices outside. Note that PDS can only exist in graph containing at least 3 vertices. Thus, in this paper we will always assume that all graphs have at least $3$ vertices. Then \maxPDS{} consists in finding a \PDS{} of maximum size in a given graph. If we require that a \PDS{} induces a connected subgraph, we refer to such problem as \connmaxPDS{}. 

Let us give more details on the background of the problem. In 2013, Olsen (see~\cite{olsenGenView}) introduced the so-called community structure problem which is defined as the problem of partitioning vertices of a given graph into at least two sets (called \textit{communities}), each containing at least two vertices, such that every vertex is adjacent to proportionally at least as many vertices in its own community as to vertices in any other community.
Notice that in the definition introduced by Olsen in~\cite{olsenGenView}, the exact number of communities is not given, i.e.\ the only restriction is that there are at least two communities.
In~\cite{bazgan:k-comm, Estivill1}, the notion of \textit{$k$-community structure} was first used to fix the number of communities to some integer $k\geq 2$, i.e.\ a $k$-community structure is a community structure containing exactly $k$ communities. This problem has been intensively studied as well (see for instance \cite{bazgan:k-comm,bazgan:firstInfiniteFamily,Estivill1,BaghirovaK-Community}. There exist other closely related problems in the literature. For instance, the \textsc{Satisfactory Partition} problem introduced in~\cite{SatPartition} consists in deciding if a given graph has a partition of its vertex set into two nonempty parts such that each vertex has at least as many neighbors in its part as in the other part. This problem and its variants have been intensively studied (see for instance~\cite{NoteOnSatPart,BalSatPart,SatPartParamComp,SatPartAlgApproach}). 

In the original definition of community structure introduced by Olsen, and the one of $k$-community structure introduced in~\cite{bazgan:k-comm, Estivill1}, each community must contain at least two vertices. One may relax this constraint and only require a community to be non-empty, i.e.\ to contain at least one vertex.
In \cite{BaghirovaK-Community}, the authors use this version as well and call such a partition a \textit{generalized $k$-community structure}.
Subsequently, based on the concept of $k$-community structures, the authors in \cite{MR4023158} introduced the notion of a  \textit{proportionally dense subgraph (PDS)} and studied the problem of finding a \PDS{} of maximum size. While closely linked to the definition of a $k$-community structure, this definition shifts from a partitioning problem to a maximization problem. Notice, that this is exactly the focus of the study in this paper. Then, in~\cite{bazgan:firstInfiniteFamily}, the authors introduce the notion of \textit{$2$-\PDS{} partitions}, which correspond exactly to generalized $2$-community structures. 


In \cite{MR4023158}, the authors show that \maxPDS{} is APX-hard on split graphs and \NP{}-hard on bipartite graphs. In the same paper, the authors also show that deciding if a \PDS{} is inclusion-wise maximal is co-\NP{}-complete on bipartite graphs. Also, the authors present a polynomial-time $(2-\frac{2}{\Delta +1})$-approximation algorithm for the problem, where $\Delta$ is the maximum degree of the graph. Finally, the authors show that the \PDS{} of maximum size can be found in linear time in Hamiltonian cubic graphs if a Hamiltonian cycle is given in the input. Recently, the authors of~\cite{LCP} introduced a framework, which allows solving a special family of partitioning problems in polynomial time in classes of graphs of bounded clique-width. As an application, they showed that the problem of finding a \PDS{} and \connmaxPDS{} of maximum size can be solved in polynomial time in classes of graphs of bounded cliquewidth. 

Finally, we would like to mention that from a theoretical perspective, it is interesting to study problems whose definition encompasses both global and local properties as in \maxPDS{}. This unique paradigm is not very common in the field of graph theory.

\subsection{Overview of the paper}
\label{subsec:ourcontribution}
In this subsection, we discuss our contribution to the study of \maxPDS{} and \connmaxPDS{} and introduce the structure of the paper.  

First, in \cref{sec:preliminaries}, we formally introduce \maxPDS{}, \connmaxPDS{} and give some necessary definitions. As mentioned above, in \cite{MR4023158} the authors showed that \maxPDS{} can be solved in linear time on Hamiltonian cubic graphs if a Hamiltonian cycle is part of the input. However, the complexity of the problem on cubic graphs remains an open question. In this paper, we study the complexity of \maxPDS{} parameterized by $\Delta$, where $\Delta$ is the maximum degree of a graph. In \cref{subsec:d and delta param}, we show that \maxPDS{} is para-\NP{}-hard parameterized by $\Delta$ and $\degen{}$, where $\degen{}$ is degeneracy of a graph. More specifically, we show that \maxPDS{} is \NP{}-hard on graphs with $\Delta=4$ and $\degen =2$, which implies \NP{}-hardness when $h=4$, where $h$ is the $h$-index of a graph. Then, in \cref{subsec:NPhardondense graphs}, we show that \maxPDS{} is \NP{}-hard when restricted to dense graphs, more specifically graphs $G$, such that $\Delta(\overline{G})\le 6$. In the same subsection, we show that on graphs $G$, such that $\degen(\overline{G})=2$ and $\overline{G}$ is bipartite, \maxPDS{} remains \NP{}-hard. Then, we note that this result implies para-\NP{}-hardness result parameterized by the clique cover number of a given graph. For each result, we consider \connmaxPDS{} and explain how the result can be extended when we require connectivity.

In \cref{sec:polynsolvablecases} we present polynomial-time solvable cases. More specifically, in \cref{subsec:h=2} we present a polynomial-time algorithm for solving \maxPDS{} on graphs with $h \leq 2$. Finally, in \cref{subsec:graphsofdegreeatleastN-3}, we show how to find a \PDS{} of maximum size in graphs $G$ such that $h(\overline{G})\le 2$. This result implies polynomial-time complexity on graphs with $n$ vertices of minimum degree at least $(n-3)$. We extend each aforementioned result on \connmaxPDS{}.

In \cref{table:results}, we summarize the parameterized complexity results for \maxPDS{} and \connmaxPDS{} that we present in this paper.
Notice that since $\degen \leq h \leq \Delta$, if there exists a polynomial-time algorithm for a parameter for solving \maxPDS{} on a corresponding entry in the table, it implies polynomial-time complexity for the entries above and on the right. For example, a polynomial time algorithm for $h \leq 2$ implies polynomial-time complexity for $h \leq 1$ and $\Delta\le 2$. Conversely, \NP{}-hardness extends to entries below and to the left.

\begin{table}[!ht]
\begin{center}
\begin{NiceTabular}{|c|ccc|}
\CodeBefore
\cellcolor{blue!15}{2-2}
\rectanglecolor{blue!15}{2-3}{3-4}
\rectanglecolor{red!15}{3-2}{5-2}
\rectanglecolor{red!15}{5-2}{5-4}
\Body
\hline
{\color{white} \huge L} & $\degen(G)$ & $h(G)$ & $\Delta(G)$ \\
\hline
1 & \textsf{P} & ~ & ~ \\
\cline{1-2}
2 & \NP{}-h & \textsf{P} & \\
\cline{1-1}\cline{3-4}
3 & {\color{red!15}$\vdots$} & $\underset{~}{?}$ & $\underset{~}{?}$ \\
\cline{1-1}\cline{3-4}
4 & ~ & ~ & \NP{}-h \\
\hline
\CodeAfter
\tikz \draw[thick] (3-|3) -- (5-|3);
\tikz \draw[thick] (4-|4) -- (5-|4);
\tikz \draw[thick] (1-|3) -- (2-|3);
\tikz \draw[thick] (1-|4) -- (2-|4);
\end{NiceTabular}
\hspace{2cm}
\begin{NiceTabular}{|c|ccc|}
\CodeBefore
\cellcolor{blue!15}{2-2}
\rectanglecolor{blue!15}{2-3}{3-4}
\rectanglecolor{red!15}{3-2}{5-2}
\rectanglecolor{red!15}{5-2}{5-4}
\Body
\hline
{\color{white} \huge L} & $\degen(\overline{G})$ & $h(\overline{G})$ & $\Delta(\overline{G})$ \\
\hline
1 & \textsf{P} & ~ & ~ \\
\cline{1-2}
2 & \NP{}-h & \textsf{P} & \\
\cline{3-4}
\vdots & ~ & $\underset{~}{?}$ & $\underset{~}{?}$ \\
\cline{3-4}
6 & ~ & ~ & \NP{}-h \\
\hline
\CodeAfter
\tikz \draw[thick] (3-|3) -- (5-|3);
\tikz \draw[thick] (4-|4) -- (5-|4);
\tikz \draw[thick] (1-|3) -- (2-|3);
\tikz \draw[thick] (1-|4) -- (2-|4);
\end{NiceTabular}
\end{center}
\caption{Summary of the complexity of \maxPDS{} and \connmaxPDS{} when parameterized by $\degen{}$, $h$ and $\Delta$ of the graph $G$ and its complement $\overline{G}$. Blue indicates polynomial-time solvability, red indicates \NP{}-hardness, and a question mark means that the complexity is unknown.}
\label{table:results}
\end{table}

\section{Preliminaries}
\label{sec:preliminaries}
\subsection{Relevant notions}
In this subsection, we define the notions used in the paper and state the assumptions made throughout the paper.

In this paper, all graphs are simple and undirected and contain at least three vertices. Let $G=(V,E)$ be a graph. A complement of a graph $G = (V,E)$, denoted by $\overline{G}$, is the graph with vertex set $V$ such that two distinct vertices are adjacent if and only if they are not adjacent in $G$. A \emph{partition} of a vertex set $V$ of a graph into $k$ sets, is a family of subsets $V_1',\ldots,V_k'$ such that $V_i'\neq \emptyset$, for all $i\in \{1,\ldots,k\}$, and every vertex is included in exactly one subset $V_i'$, for some $i\in \{1,\ldots,k\}$.
The \emph{neighborhood} of a vertex $v\in V$ is denoted by $N(v)$, its \emph{closed neighborhood} by $N[v] := N(v) \cup \{v\}$ and its \emph{degree} by $d(v):=\vert N(v)\vert$. The maximum degree of a graph $G$ is denoted by $\Delta (G)$. If the graph $G$ is clear from the context, we simply write $\Delta$. The \emph{neighborhood} of $v\in V$ \emph{in $V'\subseteq V$}, denoted by $N_{V'}(v)$, is the set of vertices that are both in $V'$ and adjacent to $v$, i.e.\ $N_{V'}(v):=N(v)\cap V'$.

The \textit{subgraph} of $G=(V,E)$ \textit{induced by $S\subseteq V$} is defined as $G[S]:=(S,\{uv \in E: u\in S \text{ and } v\in S\})$. Also, for $S \subseteq V$, we denote by $\overline{S}$ the complement of $S$ in $V$, i.e.\ $\overline{S} := V \setminus S$. We denote by $\intInterval{a,b}$, with $a,b,t \in \mathbb{N}$, the set $\{a,a+1, \ldots, b\}$.
	
A \emph{tree} is a connected, acyclic graph and a \emph{forest} is an acyclic graph. Notice, that vertices of degree $1$ are called \emph{leaves}. A \emph{star} $S_n$ is a tree on $n+1$ vertices, where exactly one vertex has degree $n$ (called the \textit{center}) and all of its $n$ neighbors are leaves.

Hereafter, we define the graph parameters mentioned in this paper. Given a graph $G$, we define the \emph{degeneracy} of $G$ as the minimum integer $\degen{(G)}$ such that every subgraph of $G$ has a minimal degree at most $\degen{(G)}$. If the graph is clear from the context, we simply write $\degen{}$. Equivalently, a graph is $k$-degenerate if its vertices can be successively deleted, according to an order, so that when deleted, each vertex has degree at most $k$. Such order is called a $k$-\emph{elimination order}. The degeneracy of a graph is the smallest $k$ such that it is $k$-degenerate.

Given a graph $G$, we define the \emph{h-index} of $G$ as the largest integer $h$ such that there are at least $h(G)$ vertices of degree at least $h(G)$. Similarly, If the graph is clear from the context, we simply write $h$.

Now, let us formally define the main concept of our paper, namely the notion of \emph{Proportionally Dense Subgraph (PDS)} (as used in~\cite{PDS}).
\begin{definition}
    Let $G=(V,E)$ be a graph and $S \subset V$, such that $2 \le \vert S \vert < \vert V \vert$. Let $\overline{S} = V \setminus S$. We say that $S$ is a \emph{Proportionally Dense Subgraph (PDS)} of $G$ if for each vertex $v\in S$, one of the following equivalent inequalities hold
    \begin{equation}
    \label{prop:PDS}
        (a)\,\,\frac{d_S(u)}{\vert S \vert - 1} \ge \frac{d_{\overline{S}}(u)}{\vert \overline{S} \vert}\,, \quad \quad  (b)\,\,\frac{d_S(u)}{\vert S \vert - 1} \ge \frac{d(u)}{ \vert V \vert -1}\,,\quad \quad (c)\,\,\frac{d(u)}{|V|-1} \geq \frac{d_{\overline{S}(u)}}{|\overline{S}|}.
    \end{equation}
\end{definition}

A \emph{connected PDS} is a \PDS in which the vertices induce a connected subgraph.

We say that a vertex $v \in S$ is \emph{satisfied with respect to $S$} if it satisfies \eqref{prop:PDS}. If the \PDS{} $S$ is clear from the context, we may simply say that a vertex $v \in S$ is \emph{satisfied}. Let $G=(V,E)$ be a graph and let $C$ be a connected component of $G$. Then, when we say that $C \subseteq S$ when every vertex of $G[C]$ is in $S$.

\maxPDS{} is defined as follows:

\noindent\rule[0.5ex]{\linewidth}{1pt}
{\maxPDS{}:}\\
\noindent\rule[0.5ex]{\linewidth}{0.5pt}
{\em Input:}\hspace{0.47cm} A graph $G$.\\
{\em Output:}\hspace{0.2cm} A \PDS{} of a maximum size in $G$.\\
\noindent\rule[0.5ex]{\linewidth}{1pt}\\

The corresponding decision problem is, given a graph $G=(V,E)$ and an integer $k$, to determine whether $G$ contains a \PDS{} of size at least $k$.

Below we define \connmaxPDS{}. 

\noindent\rule[0.5ex]{\linewidth}{1pt}
{\connmaxPDS{}:}\\
\noindent\rule[0.5ex]{\linewidth}{0.5pt}
{\em Input:}\hspace{0.47cm} A graph $G$.\\
{\em Output:}\hspace{0.2cm} A connected \PDS{} of a maximum size in $G$.\\
\noindent\rule[0.5ex]{\linewidth}{1pt}\\



The authors in \cite{MR4023158} showed the following two results, which we frequently refer to in this paper. 

\begin{theorem}\cite[Theorem 4]{MR4023158}
    \label{lowerboundS}
    For any given graph $G=(V,E)$, a proportionally dense subgraph of size $\lceil \frac{|V|}{2} \rceil$ or $\lceil \frac{|V|}{2} \rceil +1$ can be found in linear time.  
\end{theorem}



\subsection{Properties}
\label{subsec:properties}

In this subsection, we derive properties of PDSs that we use later in the paper. 
Let $G=(V,E)$ be a graph and let $S$ be a \PDS{} in $G$. 

Let us note the following. Let $v\in V$ be such that $d(v)=0$. Then, it is easy to see that $S= V\setminus\{v\}$ forms a \PDS{} of a maximum size in $G$. Hence, as of now, we assume that $d(v)\ge 1$, for all $v\in V$. 



In the lemma below, we show that if a \PDS{} $S$ in $G$ is sufficiently large and the degree of some vertex $u$ is relatively small, then the inclusion of vertex $u$ in $S$ implies the inclusion of its neighborhood $N(u)$ in $S$. 

\begin{lemma}\label{lemma:one_to_all}
Let $G=(V,E)$ be a graph and $S$ a \PDS{} in $G$ and let $u \in S$. Then if $\displaystyle d(u) < \frac{|V|-1}{|V|-|S|}$, or equivalently if $\displaystyle \frac{d(u)-1}{|S|-1} < \frac{1}{|V|-|S|}$, then we have $N(u) \subseteq S$.
\end{lemma}

\begin{proof}
Since, by assumption, $S$ is \PDS{} in $G$, we know that every $v\in V$ is satisfied with respect to $S$. Let $u\in V$ be such that $ d(u) < \frac{|V|-1}{|V|-|S|}$. Since $u$ is satisfied, we have
\begin{align*}
    & \frac{d(u)}{|V|-1} \geq \frac{d_{\overline{S}}(u)}{|V|-|S|} 
    \iff  d_{\overline{S}}(u) \leq \frac{|V|-|S|}{|V|-1} d(u) < 1. 
\end{align*}
Hence, $d_{\overline{S}}(u)=0$ and the lemma follows.
\end{proof}

The lemma introduced above implies the following.

\begin{corollary}\label{cor:one_to_all}
Let $G=(V,E)$ be a graph, $S$ be a \PDS{} in $G$, $X \subseteq V$ be such that $G[X]$ is connected and for all $u \in X$, $d(u) < \frac{|V|-1}{|V|-|S|}$ for all $u \in X$. Then, either $N[X] \subseteq S$ or $N[X] \subseteq \overline{S}$. 
\end{corollary}

\begin{proof}
If there exists $u \in X$ such that $u \in S$, then using \cref{lemma:one_to_all} it holds that all neighbors of $u$ are also in $S$. Since $G[X]$ is connected we know that by iterating this argument we will obtain $X \subseteq S$ and all neighbors of elements of $X$ are also in $S$.
\end{proof}

As mentioned earlier, a \PDS{} of size equal to or greater than $\frac{|V|}{2}$ can be found in any graph $G=(V,E)$ in linear time (see \cref{lowerboundS}). Hence, in this paper, we focus on finding a \PDS{} of size of at least $\frac{|V|}{2}+1$. Then, \cref{lemma:one_to_all} implies that in this case, vertices of degree $2$ in a \PDS{} are satisfied if and only if both of its neighbors are contained in the \PDS{} as well. 

Furthermore, as noted by Pontoizeau in \cite[Page 102]{pontoizeau}, there are examples of graphs where communities of a maximum size are not connected. More specifically, in Figure 5.3, the maximum community is of size $n-3$, while a connected community of a maximum size contains at most $\frac{n}{2}$ vertices, where $n$ is the number of vertices in the graph. This indicates that the sizes of a \PDS{} and a connected \PDS{} differ significantly in certain graphs, more specifically, by a ratio of nearly $\frac{1}{2}$. Thus, in this paper, we also study \connmaxPDS{} as well and explain how to extend each result presented in this paper when connectivity is required.

\section{Parameterized hardness}
\label{sec:parametrizedcomplex}

In this section, we introduce the following results. In \cref{subsec:d and delta param}, we show that \maxPDS{} is para-\NP{}-hard parameterized by the $\Delta$ and $\degen{}$. More specifically, we show that \maxPDS{} is \NP{}-hard on graphs with $\Delta=4$ and $\degen =2$, which implies \NP{}-hardness result on $h=4$. Further, in \cref{subsec:NPhardondense graphs}, we show that \maxPDS{} is \NP{}-hard when restricted to dense graphs. More specifically, we prove that \maxPDS{} is \NP{}-hard on graphs $G$ such that $\Delta(\overline{G}) = 6$. In the same subsection, we show that on graphs $G$, such that $\degen(\overline{G})=2$ and $\overline{G}$ is bipartite, \maxPDS{} remains \NP{}-hard.
 
\subsection{Maximum degree and degeneracy}
\label{subsec:d and delta param}
In this subsection, we establish the para-\NP{}-hardness results of \maxPDS{}, parameterized by the degeneracy and maximum degree. More specifically, we prove that \maxPDS{} is \NP{}-hard on 2-degenerate graphs where each vertex has a degree at most $4$.  Before we move on to the main results of this subsection, we introduce some preliminary results. 

\begin{lemma}
\label{variablesNKAB}
    Given a graph $G=(V,E)$, with $\vert V \vert =n$ and $\vert E \vert =m$ (notice that since graph is simple, we have $m\le n^2$), without isolated vertices and a positive integer $k$ such that $3\le k <n$, there exist integers $A,B,N$, and $k'$ such that:
  \begin{AutoMultiColItemize}
        \item[1.]\label{item1'} $ A>2m$,
        \item[2.]\label{item2'} $B>n$,
        \item[3.]\label{item4'} $k'+A>N$,
        \item[4.]\label{item5'} $\frac{2}{k'-1} < \frac{1}{N-k'}$,
        \item[5.]\label{item6'} $\frac{3}{k'-1} \ge \frac{1}{N-k'}$,
        \item[6.]\label{item7'} $\frac{3}{k'} < \frac{1}{N-k'-1}$,
        \item[7.]\label{item8'} $k' = A + m + (n-k)$,
        \item[8.]\label{item9'} $N= A + B + m + n$,
        \item[9.]\label{item10'} $N$ polynomial in $n$.
  \end{AutoMultiColItemize}
\end{lemma}

\begin{proof}
    Let $C,N,k',A,B$ be integers such that:
    \begin{itemize}
        \item $C = \max\{4(n+3m), 8n+1, 2n+2m\}$
        \item $N$ is the smallest integer such that $N> C$ and $3N+1$ is a multiple of $4$.
        \item $k' = \frac{3N+1}{4}$,
        \item $A = k' - m - (n-k)$,
        \item $B = \frac{N-1}{4} - k$.
    \end{itemize}

    \begin{enumerate}
        \item[1.] Since $N > 4(n+3m)$, we have $A \ge \frac{3}{4}N -m-n > 2m$.
        
        \item[2.] Since $N > 8n+1$, we have $B = \frac{N-1}{4} - k > 2n-k \geq n$.
        
        \item[3.] Since $N>2n+2m$ we have $k'+A = 2k'-m-(n-k) \geq \frac{3}{2}N-m-n > N$.

        \item[4.] We have $\frac{2}{k'-1} < \frac{1}{N-k'} \iff 2N < 3\frac{3N+1}{4}-1 \iff N > 1$.
        
        \item[5.] and 6. Since $k' = \frac{3N+1}{4}$, we have  $\frac{3}{k'-1} = \frac{1}{N-k'}$ and $\frac{3}{k'} < \frac{1}{N-k'-1}$.
        
        \item[7.] By definition we have $A = k'-m-(n-k)$.
        
        \item[8.] By definition we have $A+B+m+n = \frac{3N+1}{4} - m -(n-k) + \frac{N-1}{4} - k = N$.
        
        \item[9.] Since at least one of the four integers $C + i$, for $1 \leq i \leq 4$ verifies that $3(C+i)+1$ is a multiple of $4$, it holds that $N < C+5$ and $N$ is polynomial in $n$.
    \end{enumerate}
Hence, the statement of the lemma follows.
\end{proof}

Further, we move on to the main result of this section.

\begin{theorem}
    \label{theorem:NPhardDegen2andDelta4}
    \maxPDS{} is \NP{}-hard on graphs with $\degen=2$ and $\Delta =4$.
\end{theorem}

\begin{proof}
We provide a polynomial-time reduction from the \textsc{Independent Set} on cubic graphs. This problem is known to be \NP{}-hard \cite{indep_cubic}.

\medbreak
\noindent\textbf{Construction: } Given an instance of the \textsc{Independent Set}, that is, a cubic graph $G = (V,E)$ and a positive integer $k$, we construct $G' =(V',E')$, a positive integer $k'$ an instance of \maxPDS{} as follows.
\begin{itemize}
    \item The vertex set $V'$ is constructed as follows.
    \begin{itemize}
        \item We create a copy of vertices in $V$. We denote by $V_G$ the set formed by such vertices.
        
        \item We create a vertex $e_{uv}$ for each $uv$ in $E$. We denote the corresponding set of vertices by $V_E$.
        
        \item We create two sets of $A$ and $B$ new vertices, which we denote by $V_A$ and $V_B$, respectively, where integers $A$ and $B$ satisfy the inequalities from \cref{variablesNKAB}.
    \end{itemize}
    
    \item The edge set $E'$ is constructed as follows.
    \begin{itemize}
        \item Recall, that by \cref{variablesNKAB}, we have that $B>n$. For every $u\in V_G$, there exists $u' \in V_B$ such that $uu'\in E'$.
        
        \item By \cref{variablesNKAB}, we also have that $A > 2m$. Each vertex in $V_E$ is adjacent to two unique vertices in $V_A$, i.e.\ for every $u \in V_E$ there exists two vertices $w,w'\in V_A$ such that $wu, w'u \in E'$.
        
        \item Each vertex $e_{uv} \in V_E$ is adjacent to $u$ and $v$ in the set $V$.
        
        \item The vertices of set $V_A$ form a path.
        
        \item The vertices of set $V_B$ form a path.
    \end{itemize}
\end{itemize}

Finally, let $k' = A + m + (n - k)$. See \cref{fig:construction2} for an illustration of the construction. Also, since $N$ is polynomial in $N$, then $G'$ can obviously be constructed in polynomial time in $n$.

\begin{figure}[ht!]
\begin{center}
\begin{tikzpicture}[scale=0.43]

\node at (-6,4) {\underline{The corresponding graph $G'$}}; 
\node at (-26,-2) {\underline{A graph $G$}}; 
\draw[line width=0.7mm,->] (-22.5,-6) -- (-17.5,-6);

\node at (-26,-6) {
\begin{tikzpicture}[scale=0.40]
\draw[thick] (-1.0,-3.0)--(-1.0,-1.0);
\draw[thick] (-1.0,-3.0)--(2.0,-2.0);
\draw[thick] (-1.0,-3.0)--(-3.0,1.0);
\draw[thick] (-1.0,-1.0)--(-1.0,1.0);
\draw[thick] (-1.0,-1.0)--(2.0,-2.0);
\draw[thick] (-1.0,1.0)--(1.0,0.0);
\draw[thick] (-1.0,1.0)--(-3.0,1.0);
\draw[thick] (1.0,0.0)--(2.0,3.0);
\draw[thick] (-3.0,1.0)--(2.0,3.0);
\draw[thick] (2.0,3.0)--(3.0,0.8);%
\draw[thick] (2.0,-2.0)--(3.0,0.8);%
\draw[thick] (1.0,0.0)--(3.0,0.8);%

\filldraw[blue] (-1.0,-3.0) circle (5 pt);
\filldraw[blue] (-1.0,-1.0) circle (5 pt);
\filldraw[blue] (-1.0,1.0) circle (5 pt);
\filldraw[blue] (2.0,-2.0) circle (5 pt);
\filldraw[blue] (1.0,0.0) circle (5 pt);
\filldraw[blue] (-3.0,1.0) circle (5 pt);
\filldraw[blue] (2.0,3.0) circle (5 pt);
\filldraw[blue] (3.0,0.8) circle (5 pt);%

\end{tikzpicture}};

\draw[thick,violet] (-10.8,0.0)--(-11.1,1.0);
\draw[thick,violet] (-10.8,0.0)--(-11.1,-1.0);
\draw[thick,violet] (-11.1,1.0)--(-11.9,1.8);
\draw[thick,violet] (-11.9,1.8)--(-13.0,2.3);
\draw[thick,violet] (-13.0,2.3)--(-14.3,2.4);
\draw[thick,violet] (-14.3,2.4)--(-15.6,2.1);
\draw[thick,violet] (-15.6,2.1)--(-16.6,1.4);
\draw[thick,violet] (-16.6,1.4)--(-17.1,0.5);
\draw[thick,violet] (-17.1,-0.5)--(-16.6,-1.4);
\draw[thick,violet] (-16.6,-1.4)--(-15.6,-2.1);
\draw[thick,violet] (-15.6,-2.1)--(-14.3,-2.4);
\draw[thick,violet] (-14.3,-2.4)--(-13.0,-2.3);
\draw[thick,violet] (-13.0,-2.3)--(-11.9,-1.8);
\draw[thick,violet] (-11.9,-1.8)--(-11.1,-1.0);
\draw[thick,vert] (-10.8,-14.0)--(-11.0,-13.3);
\draw[thick,vert] (-10.8,-14.0)--(-11.0,-14.7);
\draw[thick,vert] (-11.0,-13.3)--(-11.4,-12.6);
\draw[thick,vert] (-11.4,-12.6)--(-12.1,-12.1);
\draw[thick,vert] (-12.1,-12.1)--(-13.0,-11.7);
\draw[thick,vert] (-13.0,-11.7)--(-14.0,-11.6);
\draw[thick,vert] (-14.0,-11.6)--(-15.0,-11.7);
\draw[thick,vert] (-15.0,-11.7)--(-15.9,-12.1);
\draw[thick,vert] (-15.9,-12.1)--(-16.6,-12.6);
\draw[thick,vert] (-16.6,-12.6)--(-17.0,-13.3);
\draw[thick,vert] (-17.0,-14.7)--(-16.6,-15.4);
\draw[thick,vert] (-16.6,-15.4)--(-15.9,-15.9);
\draw[thick,vert] (-15.9,-15.9)--(-15.0,-16.3);
\draw[thick,vert] (-15.0,-16.3)--(-14.0,-16.4);
\draw[thick,vert] (-14.0,-16.4)--(-13.0,-16.3);
\draw[thick,vert] (-13.0,-16.3)--(-12.1,-15.9);
\draw[thick,vert] (-12.1,-15.9)--(-11.4,-15.4);
\draw[thick,vert] (-11.4,-15.4)--(-11.0,-14.7);

\filldraw[red] (1.8,-14.0) circle (3 pt);
\filldraw[red] (1.5,-13.2) circle (3 pt);
\filldraw[red] (1.0,-12.6) circle (3 pt);
\filldraw[red] (0.2,-12.3) circle (3 pt);
\filldraw[red] (-0.6,-12.4) circle (3 pt);
\filldraw[red] (-1.3,-12.8) circle (3 pt);
\filldraw[red] (-1.7,-13.6) circle (3 pt);
\filldraw[red] (-1.7,-14.4) circle (3 pt);
\filldraw[red] (-1.3,-15.2) circle (3 pt);
\filldraw[red] (-0.6,-15.6) circle (3 pt);
\filldraw[red] (0.2,-15.7) circle (3 pt);
\filldraw[red] (1.0,-15.4) circle (3 pt);
\filldraw[red] (1.5,-14.8) circle (3 pt);
\filldraw[violet] (-10.8,0.0) circle (3 pt);
\filldraw[violet] (-11.1,1.0) circle (3 pt);
\filldraw[violet] (-11.9,1.8) circle (3 pt);
\filldraw[violet] (-13.0,2.3) circle (3 pt);
\filldraw[violet] (-14.3,2.4) circle (3 pt);
\filldraw[violet] (-15.6,2.1) circle (3 pt);
\filldraw[violet] (-16.6,1.4) circle (3 pt);
\filldraw[violet] (-17.1,0.5) circle (3 pt);
\filldraw[violet] (-17.1,-0.5) circle (3 pt);
\filldraw[violet] (-16.6,-1.4) circle (3 pt);
\filldraw[violet] (-15.6,-2.1) circle (3 pt);
\filldraw[violet] (-14.3,-2.4) circle (3 pt);
\filldraw[violet] (-13.0,-2.3) circle (3 pt);
\filldraw[violet] (-11.9,-1.8) circle (3 pt);
\filldraw[violet] (-11.1,-1.0) circle (3 pt);
\filldraw[vert] (-10.8,-14.0) circle (3 pt);
\filldraw[vert] (-11.0,-13.3) circle (3 pt);
\filldraw[vert] (-11.4,-12.6) circle (3 pt);
\filldraw[vert] (-12.1,-12.1) circle (3 pt);
\filldraw[vert] (-13.0,-11.7) circle (3 pt);
\filldraw[vert] (-14.0,-11.6) circle (3 pt);
\filldraw[vert] (-15.0,-11.7) circle (3 pt);
\filldraw[vert] (-15.9,-12.1) circle (3 pt);
\filldraw[vert] (-16.6,-12.6) circle (3 pt);
\filldraw[vert] (-17.0,-13.3) circle (3 pt);
\filldraw[vert] (-17.0,-14.7) circle (3 pt);
\filldraw[vert] (-16.6,-15.4) circle (3 pt);
\filldraw[vert] (-15.9,-15.9) circle (3 pt);
\filldraw[vert] (-15.0,-16.3) circle (3 pt);
\filldraw[vert] (-14.0,-16.4) circle (3 pt);
\filldraw[vert] (-13.0,-16.3) circle (3 pt);
\filldraw[vert] (-12.1,-15.9) circle (3 pt);
\filldraw[vert] (-11.4,-15.4) circle (3 pt);
\filldraw[vert] (-11.0,-14.7) circle (3 pt);

\draw[-{Straight Barb[left]},line width=0.5mm,black] (-2.5,-0.25)--(-10.0,-0.25);
\node[rotate=180] at (-6.2,-0.9) {$1$};
\draw[-{Straight Barb[left]},line width=0.5mm,black] (-10.0,0.25)--(-2.5,0.25);
\node[rotate=0] at (-6.2,0.9) {$1$};

\draw[-{Straight Barb[left]},line width=0.5mm,black] (0.25,-2.9)--(0.25,-11.5);
\node[rotate=-90] at (0.9,-7.3) {$3$};
\draw[-{Straight Barb[left]},line width=0.5mm,black] (-0.25,-11.5)--(-0.25,-2.9);
\node[rotate=90] at (-0.9,-7.3) {$2$};

\draw[-{Straight Barb[left]},line width=0.5mm,black] (-2.5,-14.25)--(-10.0,-14.25);
\node[rotate=180] at (-6.2,-14.9) {$2$};
\draw[-{Straight Barb[left]},line width=0.5mm,black] (-10.0,-13.75)--(-2.5,-13.75);
\node[rotate=0] at (-6.2,-13.1) {$1$};

\draw[thick,blue,dotted] (-0.5,-1.5)--(-0.5,-0.5);
\draw[thick,blue,dotted] (-0.5,-1.5)--(1.0,-1.0);
\draw[thick,blue,dotted] (-0.5,-1.5)--(-1.5,0.5);
\draw[thick,blue,dotted] (-0.5,-0.5)--(-0.5,0.5);
\draw[thick,blue,dotted] (-0.5,-0.5)--(1.0,-1.0);
\draw[thick,blue,dotted] (-0.5,0.5)--(0.5,0.0);
\draw[thick,blue,dotted] (-0.5,0.5)--(-1.5,0.5);
\draw[thick,blue,dotted] (0.5,0.0)--(1.0,1.5);
\draw[thick,blue,dotted] (-1.5,0.5)--(1.0,1.5);
\draw[thick,blue,dotted] (1.0,1.5)--(1.5,0.4);%
\draw[thick,blue,dotted] (1.0,-1.0)--(1.5,0.4);%
\draw[thick,blue,dotted] (0.5,0.0)--(1.5,0.4);%

\draw[thick,blue] (0.0,0.0) ellipse (2.5cm and 2.9cm);
\node[blue] at (4,3.8) {$V_G$};
\node[blue] at (4,2.8) {$|V_G|=n$};

\draw[thick,red] (0.0,-14.0) ellipse (2.5cm and 2.5cm);
\node[red] at (3.8,-15.8) {$V_E,$};
\node[red] at (3.8,-16.8) {$|V_E|=m$};

\draw[thick,violet] (-14.0,0.0) ellipse (4.0cm and 3.0cm);
\node[violet] at (-18.5,3.6) {$V_B,$};
\node[violet] at (-18.5,2.6) {$|V_B|=B$};

\draw[thick,vert] (-14.0,-14.0) ellipse (4.0cm and 3.0cm);
\node[vert] at (-18.5,-15.9) {$V_{A},$};
\node[vert] at (-18.5,-16.9) {$|V_{A}|=A$};

\filldraw[blue] (-0.5,-1.5) circle (3 pt);
\filldraw[blue] (-0.5,-0.5) circle (3 pt);
\filldraw[blue] (-0.5,0.5) circle (3 pt);
\filldraw[blue] (1.0,-1.0) circle (3 pt);
\filldraw[blue] (0.5,0.0) circle (3 pt);
\filldraw[blue] (-1.5,0.5) circle (3 pt);
\filldraw[blue] (1.0,1.5) circle (3 pt);
\filldraw[blue] (1.5,0.4) circle (3 pt);

\draw[thick,violet,dashed] (-11.3,2.3) arc [start angle=157.11, end angle=202.89,x radius=5.86cm, y radius=5.86cm];
\draw[thick,vert,dashed] (-11.4,-11.72) arc [start angle=157.11, end angle=202.89,x radius=5.86cm, y radius=5.86cm];
\end{tikzpicture}
\end{center}

\caption{On the left we present an example of a cubic graph $G$ on $8$ vertices. On the right, we present the construction explained above. A double arrow $X \overset{x}{\underset{y}{\rightleftharpoons}} Y$ represents that the vertices in $X$ have at most $x$ neighbors in $Y$ and the vertices in $Y$ have at most $y$ neighbors in $X$.}
\label{fig:construction2}
\end{figure}
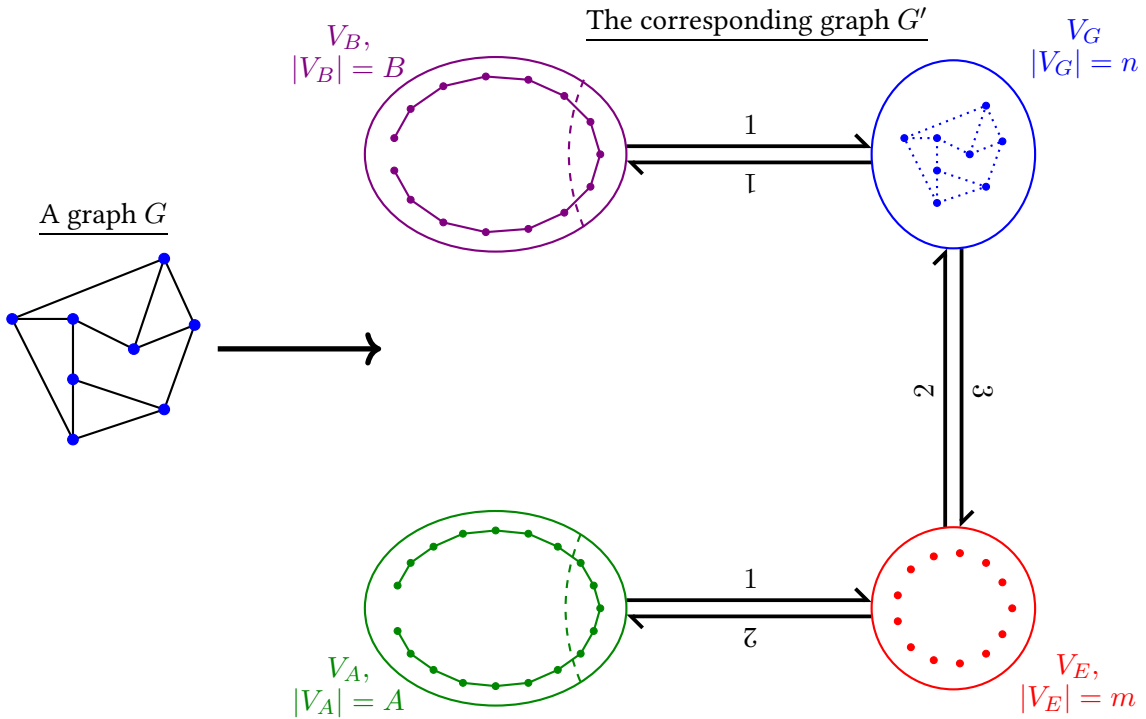

\medbreak

\noindent\textbf{Parameters $\degen$ and $\Delta$: }
Hereafter we justify that $\Delta(G')= 4$ and $\degen(G') =2$. Indeed, notice that since the vertices of the set $V_A$ (resp.\ $V_B$) form a path, each vertex has two neighbors in $V_A$ (resp.\ $V_B$) and at most one neighbor in $V_E$ (resp.\ $V_G$) by construction.
Each vertex of the set $V_G$ has one neighbor in $V_B$, by definition, and since $G$ is a cubic graph, we know that each vertex has three incident edges and hence every $u\in V_G$ has exactly three neighbors in $V_E$. Then, each vertex of $V_G$ has degree $4$.
Finally, each vertex of the set $V_E$ is adjacent to $2$ vertices of $V_G$ and $2$ vertices of $V_A$, by construction and hence each vertex has degree $4$. We conclude, that $\Delta(G') =4$. 

In order to prove that $\degen(G')=2$, we provide an elimination order of the vertices that only eliminates vertices of degree at most 2 at each step.
Since $A>2m$ (resp.\ $B>n$), there exists at least one vertex of degree $2$ in $V_A$ (resp.\ $V_B$).
We start by eliminating these vertices first. As long as there are vertices left in $V_A$ or $V_B$ there is at least one vertex of degree $2$ that we can eliminate. Indeed, any remaining vertex in $V_A$ or $V_B$ adjacent to an eliminated vertex has a degree at most 2.
Thus, we can eliminate all vertices of $V_A$ and $V_B$.
Then, after eliminating vertices in the sets $V_A$ and $V_B$, by construction vertices of $V_E$ have degree 2.
We can eliminate all of them. The remaining vertices are the vertices of $V_G$, which now form an independent set, thus we can eliminate all of them. Hence we conclude that $d(G') =2$.

\medbreak
\noindent\textbf{Soundness: }
We show that a graph $G=(V,E)$ has an independent set of size at least $k$ if and only if the graph $G' = (V',E')$, constructed as explained above, contains a \PDS{} of size at least $k'$.

\begin{itemize}

\item[$\implies$] First, assume that $I$ is an independent set of size $k$ in $G$. Let $I'$ be the set of vertices in $G'$ corresponding to the set $I$. Let $S = V_A \cup V_E \cup (V_G \setminus I')$. Hereafter, we show that $S$ is a \PDS{} of size $k'$ in $G'$. Notice, that $\vert S \vert = \vert V_A \vert + \vert V_E \vert + \vert(V_G \setminus I')\vert = A + m + (n-k) = k'$. Now, we have to show that all the vertices are satisfied with respect to $S$.

\begin{itemize}
    \item Let $u \in V_A$. Notice, that all neighbors of $u$ are in $S$, and hence $u$ is satisfied with respect to $S$.

    \item Let $e_{uv} \in V_E$. Notice, that since $I$ is an independent set and $u,v$ are endpoints of the same edge, we have that if $u\in I$, then $v\notin I$ and vice versa. Since, by the previous argument, one of the neighbors of the vertex $e_{uv}$ is not in $I$, we know that $e_{uv}$ has at most one neighbor outside of $S$. By \cref{item6'}, Item 5 we have the following:
    $$ \frac{d_{S}(e_{uv})}{|S|-1} \geq \frac{3}{k'-1} \geq \frac{1}{N-k'} \geq \frac{d_{\overline{S}}(u)}{|\overline{S}|}. $$
    
    Hence, $e_{uv}$ is satisfied.
    
    \item Let $u \in V \cap S$. We have that $u$ has 3 neighbors in $S$ and $1$ neighbor in $\overline{S}$. Thus, by applying the \cref{item6'}, Item 5,
    $$ \frac{3}{k'-1} \geq \frac{1}{N-k'}. $$
    Hence, $u$ is satisfied.
\end{itemize}
Hence, all vertices are satisfied with respect to $S$ and we conclude that $S$ is indeed a \PDS{} of size $k'$ in $G'$.

\item[$\impliedby$] Now, assume that $S$ is a \PDS{} of size at least $k'$ in $G'$. Let $I$ be a set of vertices in $G$ corresponding to the set $I' = V'\cap (V_G \setminus S)$. Hereafter, we show that $I$ is an independent set in $G$ of size $k$. 

Notice that by \cref{variablesNKAB}, we have that $k\ge 3$ and $k'+A>N$. Moreover, since $k'=A+m+(n-k)$ and $N=A+B+m+n$, we have that $A-k>B$ and hence $A>B$.

\begin{itemize}
    \item First assume that $|S| > k'$. As mentioned above, by \cref{variablesNKAB}, we have that $A>B$ and hence we know that at least one vertex of the set $V_A$ is contained in $S$. By \cref{cor:one_to_all}, we conclude that $V_A \cap V_E \subset S$. Similarly, we conclude that $V_B \cap V_G \subset S$ and $S = V'$, a contradiction. Thus, $|S| \leq k'$ and $|S|=k'$.

    \item Similary, as mentioned above we have that $A>B$ and there exists $u \in V_A$ such that $u \in S$. By applying \cref{item5'} and \cref{cor:one_to_all} since vertices in $V_A$ have degree at most $3$, we conclude that $V_A \cup V_E \subseteq S$.
    
    \item We use a similar argument as above for the set $V_B$. Since vertices of $V_B$ have degree at most $3$, we know that $v \in V_B \cap S$, for some $v\in V_B$, then $V_B \cup V \subset S$ and $S = V'$, a contradiction. Hence, $S \cap V_B = \emptyset$. Then, $S = V_A \cup V_E \cup (V_G \setminus I')$. And since, $|S| = k' = A + m + (n-k)$, it holds that $|I|=k$.
    
    \item Assume that for some $u,v \in I$, we have that $uv \in E$. Then, we know that $e_{uv}\in V_E$ is adjacent to both $u$ and $v$. Then, since $S$ is a \PDS{} by assumption, $e_{uv}$ is satisfied with respect to $S$. Moreover, we have that $|S|=k'$ and $e_{uv}$ has degree 4. If two of its neighbors are in $\overline{S}$ we have:
    $$\frac{2}{k'-1}\ge \frac{2}{N-k'}\,,$$
    a contradiction to \cref{item5'}, Item 4.
\end{itemize}
\end{itemize}
Hence, $I$ is an independent set of size $k$ in $G$ which concludes the proof.
\end{proof}

Then, the theorem above implies the following result.

\begin{corollary}
    \maxPDS{} is para-\NP{}-hard parameterized by $\degen$ and $\Delta$ parameters.
\end{corollary}

Notice, that the \PDS{} in the proof of \cref{theorem:NPhardDegen2andDelta4} induces a connected subgraph of $G'$. This implies the following result.

\begin{corollary}
    \connmaxPDS{} is para-\NP{}-hard parameterized by $\degen$ and $\Delta$ parameters.
\end{corollary}

Notice that since $h \le \Delta$, the \NP{}-hardness on graphs with $\Delta \leq 4$ implies the \NP{}-hardness on graphs with $h \leq 4$. 

Notice, that in \cite{MR4023158} authors showed that \maxPDS{} is \NP{}-hard on bipartite graphs. However, next we highlight that we can obtain the same result by using a similar construction to the one introduced above. 

More specifically, in order to ensure that $G'$ is bipartite when constructing the edges between $V_A$ and $V_E$ (resp.\ $V_B$ and $V_G$) we create edges between the sets, connecting every second vertex from $V_A$ (resp.\ $V_B$) to $V_E$ (resp.\ $V_G$) as shown in \cref{fig:bipartite}. Since $G'$ is bipartite, we know that $G'$ is $2$-colorable. We present a $2$-coloring on the figure as well.

\begin{figure}[!h]
\centering
\begin{tikzpicture}[scale=0.5]
\draw[thick,dashed,vert] (1.2,-4.0)--(1.5,-3.2);
\draw[thick,vert] (1.5,-3.2)--(1.8,-2.2);
\draw[thick,black] (1.5,-3.2)--(8.1,-1.3);
\draw[thick,vert] (1.8,-2.2)--(1.9,-1.2);
\draw[thick,vert] (1.9,-1.2)--(2.0,0.0);
\draw[thick,black] (1.9,-1.2)--(8.1,-1.3);
\draw[thick,vert] (2.0,0.0)--(1.9,1.2);
\draw[thick,vert] (1.9,1.2)--(1.8,2.2);
\draw[thick,black] (1.9,1.2)--(8.1,1.3);
\draw[thick,vert] (1.8,2.2)--(1.5,3.2);
\draw[thick,dashed,vert] (1.5,3.2)--(1.2,4.0);
\draw[thick,black] (1.5,3.2)--(8.1,1.3);

\filldraw[vert] (1.5,-3.2) circle (4pt) node[anchor=east] {\color{black}0};
\filldraw[vert] (1.8,-2.2) circle (4pt) node[anchor=east] {\color{black}1};
\filldraw[vert] (1.9,-1.2) circle (4pt) node[anchor=east] {\color{black}0};
\filldraw[vert] (2.0,0.0) circle (4pt) node[anchor=east] {\color{black}1};
\filldraw[vert] (1.9,1.2) circle (4pt) node[anchor=east] {\color{black}0};
\filldraw[vert] (1.8,2.2) circle (4pt) node[anchor=east] {\color{black}1};
\filldraw[vert] (1.5,3.2) circle (4pt) node[anchor=east] {\color{black}0};

\filldraw[red] (10,-2.3) circle (5pt) node[anchor=south] {\color{black} 1};
\filldraw[red] (10,2.3) circle (5pt) node[anchor=north] {\color{black} 1};
\filldraw[red] (8.1,-1.3) circle (5pt) node[anchor=south] {\color{black} 1};
\filldraw[red] (8.1,1.3) circle (5pt) node[anchor=north] {\color{black} 1};
\filldraw[red] (11.9,1.3) circle (5pt) node[anchor=north] {\color{black} 1};
\filldraw[red] (11.9,-1.3) circle (5pt) node[anchor=south] {\color{black} 1};

\draw[red,thick] (10,0) ellipse (2.6cm and 2.6cm);
\node[red] () at (10,-3.3) {$V_E$};

\draw[thick,vert] (2.5,0) arc [start angle=0, end angle=50, x radius=2.6, y radius=6.5];
\draw[thick,vert] (2.5,0) arc [start angle=0, end angle=-50, x radius=2.6, y radius=6.5];
\node[vert] () at (3,-4) {$V_A$};

\draw[-{Straight Barb[left]},line width=0.5mm,black] (12.65,0.2)--(19,0.2);
\draw[-{Straight Barb[left]},line width=0.5mm,black] (19,-0.2)--(12.65,-0.2);

\node () at (25,0) {
\begin{tikzpicture}[scale=0.5,rotate=180]
\draw[thick,dashed,violet] (1.2,-4.0)--(1.5,-3.2);
\draw[thick,violet] (1.5,-3.2)--(1.8,-2.2);
\draw[thick,violet] (1.8,-2.2)--(1.9,-1.2);
\draw[thick,violet] (1.9,-1.2)--(2.0,0.0);
\draw[thick,black] (1.9,-1.2)--(8.1,-1.3);
\draw[thick,violet] (2.0,0.0)--(1.9,1.2);
\draw[thick,violet] (1.9,1.2)--(1.8,2.2);
\draw[thick,black] (1.9,1.2)--(8.1,1.3);
\draw[thick,violet] (1.8,2.2)--(1.5,3.2);
\draw[thick,dashed,violet] (1.5,3.2)--(1.2,4.0);

\draw[thick,black] (1.5,3.2) .. controls (7,4) .. (11.9,1.3);
\draw[thick,black] (1.5,-3.2) .. controls (7,-4) .. (11.9,-1.3);

\filldraw[violet] (1.5,-3.2) circle (4pt) node[anchor=west] {\color{black}1};
\filldraw[violet] (1.8,-2.2) circle (4pt) node[anchor=west] {\color{black}0};
\filldraw[violet] (1.9,-1.2) circle (4pt) node[anchor=west] {\color{black}1};
\filldraw[violet] (2.0,0.0) circle (4pt) node[anchor=west] {\color{black}0};
\filldraw[violet] (1.9,1.2) circle (4pt) node[anchor=west] {\color{black}1};
\filldraw[violet] (1.8,2.2) circle (4pt) node[anchor=west] {\color{black}0};
\filldraw[violet] (1.5,3.2) circle (4pt) node[anchor=west] {\color{black}1};

\filldraw[blue] (8.1,-1.3) circle (5pt) node[anchor=east] {\color{black} 0};
\filldraw[blue] (8.1,1.3) circle (5pt) node[anchor=east] {\color{black} 0};
\filldraw[blue] (11.9,1.3) circle (5pt) node[anchor=west] {\color{black} 0};
\filldraw[blue] (11.9,-1.3) circle (5pt) node[anchor=west] {\color{black} 0};

\draw[blue,thick] (10,0) ellipse (2.6cm and 2.6cm);

\draw[thick,violet] (2.5,0) arc [start angle=0, end angle=50, x radius=2.6, y radius=6.5];
\draw[thick,violet] (2.5,0) arc [start angle=0, end angle=-50, x radius=2.6, y radius=6.5];

\node[blue] () at (10,3.5) {$V_G$};
\node[violet] () at (3,4.5) {$V_B$};

\end{tikzpicture}
};

\end{tikzpicture}
\caption{The adjacency between the sets $V_A$ and $V_E$ (resp.\ $V_B$ and $V_G$). The $\{0,1\}$-labeling represent a 2-coloring of $G'$.}
\label{fig:bipartite}
\end{figure}
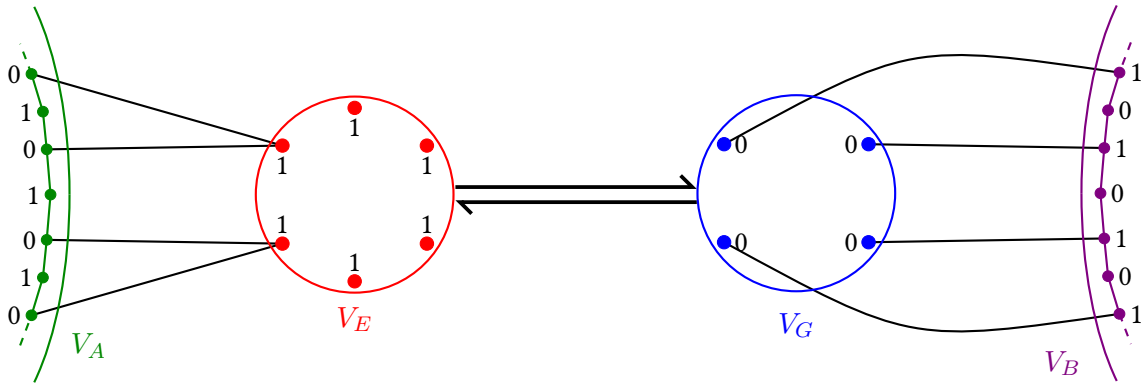

Next, we explain why employing a similar approach encounters limitations when using the same methodology on graphs with $\Delta=3$. Let $G'=(V',E')$ be a graph constructed as explained above. Let $S$ be a \PDS{} in $G'$. Intuitively, the reduction provided above relies on the following idea. We choose values of $A$ and $B$ such that $A, B>>n$ and $A \approx 3B$. We also choose the size of $S$ to be of the following value: $k' = A + m + (n-k)$. In this way, we ensure that at least one vertex of the set $V_A$ will be contained in $S$. Furthermore, $A \approx \frac{3}{4}|V'|$ and vertices contained in $V_A$ have degree at most $3$. Intuitively, considering the chosen size of $S$, we know that at least one vertex from $V_A$ is included in $S$, and \cref{lemma:one_to_all} implies that vertices of $V_A$ and $V_E$ are contained in $S$. Consequently, $V_A$ serves as a ``magnet" to ensure that $V_E\subset S$. A similar observation applies to the set $V_B$.

However, we could not extend this approach on graphs with $\Delta = 3$. If we require sets $V_A$ and $V_B$, respectively, to induce connected subgraphs such that the vertices have a degree at most $2$ (to ensure that \cref{lemma:one_to_all} implies that vertices of $V_A$ and $V_E$ are contained in $S$), then these subgraphs must be either paths or cycles. However, since degree is bounded by $2$ those subgraphs can not be connected to the rest of the graph. 

\subsection{Parameterization of the complement graph}
\label{subsec:NPhardondense graphs}
In this subsection, we prove that \maxPDS{} is \NP{}-hard when restricted to dense graphs, i.e. graphs whose complement is sparse. More specifically, we prove that \maxPDS{} is \NP{}-hard on graphs $G$ such that $\Delta(\overline{G}) = 6$. Moreover, we show that \maxPDS{} is \NP{}-hard on graphs $G$ such that $\degen(\overline{G})=2$ and $\overline{G}$ is bipartite. Before we move on to the main results of this subsection, we introduce some preliminary results.

\begin{lemma}
\label{lemma:xgeqY}
    Let $x,y >0$ be integers. Then, $\frac{x-1}{x} \geq \frac{y-1}{y}$ if and only if and $x \geq y$.
\end{lemma}

\begin{proof}
The function $x \in \ ]0,+\infty[ \ \mapsto \frac{x-1}{x}=1-\frac{1}{x}$ is strictly increasing and hence the statement follows.
\end{proof}

In the lemma below, we show that if a \PDS{} $S$ in $G$ is sufficiently large, then a vertex in $S$, which is not adjacent to exactly one vertex in $S$ and exactly one vertex in $\overline{S}$, is satisfied.

\begin{lemma}
    Let $G=(V,E)$ be a graph and $S \subseteq V$ such that $|S| > \frac{|V|}{2}$. Let $u \in S$ such that $u$ has two non-neighbors, one in $S$, the other in $\overline{S}$. Then, $u$ is satisfied with respect to $S$.
\end{lemma}

\begin{proof}
Since $|S| > \frac{|V|}{2} > |\overline{S}|$ and by \cref{lemma:xgeqY}, we have $\frac{|S|-2}{|S|-1} \geq \frac{|\overline{S}|-1}{|\overline{S}|}$.
\end{proof}

In the proof of \cref{thm:MaxPDS_NPh_Delta_complement}, we start with a graph $G$ and transform it into a graph $G'=(V', E')$. See \cref{fig:delta_compl_1,fig:delta_compl_2} for an example of such construction. The construction ensures that a vertex $e_{uv}$ is satisfied with respect to the set $S$ if and only if at least one of the vertices $u$ or $v$ is not in $S$, implying that $S$ contains an independent set of $G$. See the proof of the \cref{thm:MaxPDS_NPh_Delta_complement} for more details.



\begin{theorem}\label{thm:MaxPDS_NPh_Delta_complement}
\maxPDS{} is \NP{}-hard on graphs $G$ such that $\Delta(\overline{G})=6$. 
\end{theorem}

\begin{proof}
We provide a polynomial-time reduction from the \textsc{Independent Set} on cubic graphs. This problem is known to be \NP{}-hard \cite{indep_cubic}.

\paragraph{Construction:} Given an instance of \textsc{Independent Set}, that is, a cubic graph $G = (V,E)$, with $|V| =n$ and $|E| = m$, and a positive integer $k$, we construct $G' =(V',E')$, a positive integer $k'$ an instance of \maxPDS{} as follows. 

\begin{itemize}
    \item The vertex set $V'$ is constructed as follows.
    \begin{itemize}
        \item We create a copy of vertices in $V$. We denote by $V_G$ the set formed by such vertices. 

        \item We create a vertex $e_{uv}$ for each $uv \in E$. We denote the corresponding set of vertices by $V_E$.
    \end{itemize}

    \item The edge set $E'$ is constructed as follows.
    
 \begin{itemize}
        \item Let $v\in V$ and $w,u,x$ be its neighbors in $G$. Let $v',w',u',x'$ be the corresponding vertices in $G'$. Then, $v'$ is adjacent to all vertices other than $w',u',x',e_{vw},e_{vu}$, and $e_{vx}$ in $G'$. Notice, that since $G$ is a cubic graph, $v$ is not adjacent to four vertices in $V_G$ and $3$ vertices in $V_E$.
        
        \item The vertices of $V_E$ form a clique.
    \end{itemize}
\end{itemize}

The complement of $G'$ is the graph obtained by replacing each edge $uv$ of $G$ with a triangle between vertices $u$, $v$ and $e_{uv}$. See \Cref{fig:delta_compl_1} for the illustration of the complement of $G'$ and \Cref{fig:delta_compl_2} for an illustration of the graph $G'$. Finally, let $k'=m+k$. Notice that $|V'|=m+n = \frac{5}{2}n$ since $G$ is a cubic graph and $k' \geq \frac{3}{2}n > \frac{|V'|}{2}$.

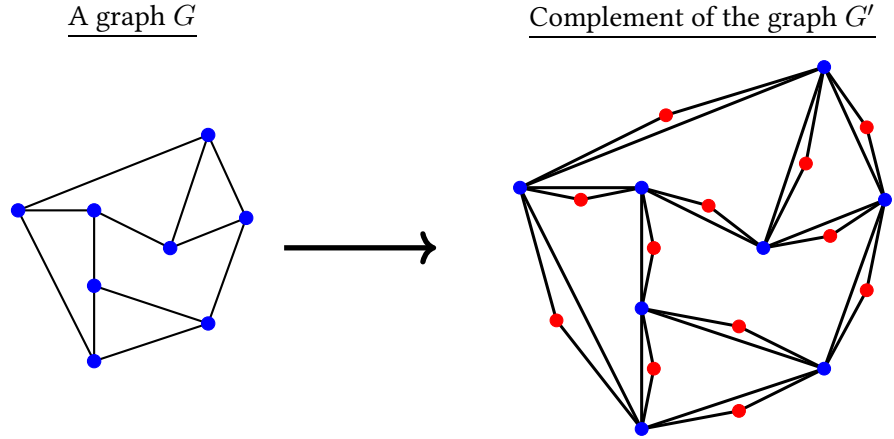
\begin{figure}[ht!]
\begin{center}
\begin{tikzpicture}[scale=0.5]

\node at (0,6) {\underline{Complement of the graph $G'$}}; 
\node at (-15,6) {\underline{A graph $G$}}; 
\draw[line width=0.7mm,->] (-11,0) -- (-7,0);

\node at (-15,0) {
\begin{tikzpicture}[scale=0.5]
\draw[thick] (-1.0,-3.0)--(-1.0,-1.0);
\draw[thick] (-1.0,-3.0)--(2.0,-2.0);
\draw[thick] (-1.0,-3.0)--(-3.0,1.0);
\draw[thick] (-1.0,-1.0)--(-1.0,1.0);
\draw[thick] (-1.0,-1.0)--(2.0,-2.0);
\draw[thick] (-1.0,1.0)--(1.0,0.0);
\draw[thick] (-1.0,1.0)--(-3.0,1.0);
\draw[thick] (1.0,0.0)--(2.0,3.0);
\draw[thick] (-3.0,1.0)--(2.0,3.0);
\draw[thick] (2.0,3.0)--(3.0,0.8);%
\draw[thick] (2.0,-2.0)--(3.0,0.8);%
\draw[thick] (1.0,0.0)--(3.0,0.8);%

\filldraw[blue] (-1.0,-3.0) circle (5 pt);
\filldraw[blue] (-1.0,-1.0) circle (5 pt);
\filldraw[blue] (-1.0,1.0) circle (5 pt);
\filldraw[blue] (2.0,-2.0) circle (5 pt);
\filldraw[blue] (1.0,0.0) circle (5 pt);
\filldraw[blue] (-3.0,1.0) circle (5 pt);
\filldraw[blue] (2.0,3.0) circle (5 pt);
\filldraw[blue] (3.0,0.8) circle (5 pt);%
\end{tikzpicture}};

\node at (0,0) {
\begin{tikzpicture}[scale=0.8]
\draw[line width=0.4mm, black] (-1.0,-3.0) -- (-1.0, -1.0);
\draw[line width=0.4mm, black] (-1.0,-3.0) -- (2.0, -2.0);
\draw[line width=0.4mm, black] (-1.0,-3.0) -- (-3.0, 1.0);
\draw[line width=0.4mm, black] (-1.0,-3.0) -- (-0.8, -2.0);
\draw[line width=0.4mm, black] (-1.0,-3.0) -- (0.6, -2.7);
\draw[line width=0.4mm, black] (-1.0,-3.0) -- (-2.4, -1.2);
\draw[line width=0.4mm, black] (-1.0,-1.0) -- (-1.0, 1.0);
\draw[line width=0.4mm, black] (-1.0,-1.0) -- (2.0, -2.0);
\draw[line width=0.4mm, black] (-1.0,-1.0) -- (-0.8, -2.0);
\draw[line width=0.4mm, black] (-1.0,-1.0) -- (-0.8, 0.0);
\draw[line width=0.4mm, black] (-1.0,-1.0) -- (0.6, -1.3);
\draw[line width=0.4mm, black] (-1.0,1.0) -- (1.0, 0.0);
\draw[line width=0.4mm, black] (-1.0,1.0) -- (-3.0, 1.0);
\draw[line width=0.4mm, black] (-1.0,1.0) -- (-0.8, 0.0);
\draw[line width=0.4mm, black] (-1.0,1.0) -- (0.1, 0.7);
\draw[line width=0.4mm, black] (-1.0,1.0) -- (-2.0, 0.8);
\draw[line width=0.4mm, black] (2.0,-2.0) -- (3.0, 0.8);
\draw[line width=0.4mm, black] (2.0,-2.0) -- (0.6, -2.7);
\draw[line width=0.4mm, black] (2.0,-2.0) -- (0.6, -1.3);
\draw[line width=0.4mm, black] (2.0,-2.0) -- (2.7, -0.7);
\draw[line width=0.4mm, black] (1.0,0.0) -- (2.0, 3.0);
\draw[line width=0.4mm, black] (1.0,0.0) -- (3.0, 0.8);
\draw[line width=0.4mm, black] (1.0,0.0) -- (0.1, 0.7);
\draw[line width=0.4mm, black] (1.0,0.0) -- (1.7, 1.4);
\draw[line width=0.4mm, black] (1.0,0.0) -- (2.1, 0.2);
\draw[line width=0.4mm, black] (-3.0,1.0) -- (2.0, 3.0);
\draw[line width=0.4mm, black] (-3.0,1.0) -- (-2.4, -1.2);
\draw[line width=0.4mm, black] (-3.0,1.0) -- (-2.0, 0.8);
\draw[line width=0.4mm, black] (-3.0,1.0) -- (-0.6, 2.2);
\draw[line width=0.4mm, black] (2.0,3.0) -- (3.0, 0.8);
\draw[line width=0.4mm, black] (2.0,3.0) -- (1.7, 1.4);
\draw[line width=0.4mm, black] (2.0,3.0) -- (-0.6, 2.2);
\draw[line width=0.4mm, black] (2.0,3.0) -- (2.7, 2.0);
\draw[line width=0.4mm, black] (3.0,0.8) -- (2.7, -0.7);
\draw[line width=0.4mm, black] (3.0,0.8) -- (2.1, 0.2);
\draw[line width=0.4mm, black] (3.0,0.8) -- (2.7, 2.0);
\filldraw[blue] (-1.0,-3.0) circle (3 pt);
\filldraw[blue] (-1.0,-1.0) circle (3 pt);
\filldraw[blue] (-1.0,1.0) circle (3 pt);
\filldraw[blue] (2.0,-2.0) circle (3 pt);
\filldraw[blue] (1.0,0.0) circle (3 pt);
\filldraw[blue] (-3.0,1.0) circle (3 pt);
\filldraw[blue] (2.0,3.0) circle (3 pt);
\filldraw[blue] (3.0,0.8) circle (3 pt);
\filldraw[red] (-0.8,-2.0) circle (3 pt);
\filldraw[red] (0.6,-2.7) circle (3 pt);
\filldraw[red] (-2.4,-1.2) circle (3 pt);
\filldraw[red] (-0.8,0.0) circle (3 pt);
\filldraw[red] (0.6,-1.3) circle (3 pt);
\filldraw[red] (0.1,0.7) circle (3 pt);
\filldraw[red] (-2.0,0.8) circle (3 pt);
\filldraw[red] (2.7,-0.7) circle (3 pt);
\filldraw[red] (1.7,1.4) circle (3 pt);
\filldraw[red] (2.1,0.2) circle (3 pt);
\filldraw[red] (-0.6,2.2) circle (3 pt);
\filldraw[red] (2.7,2.0) circle (3 pt);
\end{tikzpicture}
};

\end{tikzpicture}
\end{center}

\caption{On the left we present an example of a cubic graph $G$ on $8$ vertices. On the right, we present the complement of the graph $G'$ constructed as explained above.}
\label{fig:delta_compl_1}
\end{figure}

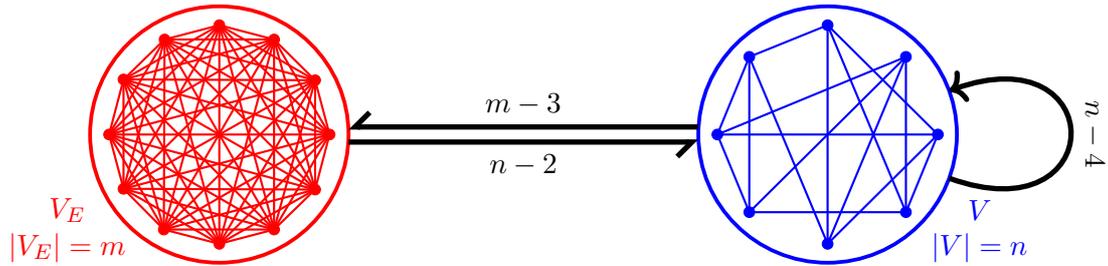
\begin{figure}[ht!]
\begin{center}
\begin{tikzpicture}[scale=1]
\draw[blue, thick] (1.45,0.00) -- (0.00, 1.45);
\draw[blue, thick] (1.45,0.00) -- (-1.45, 0.00);
\draw[blue, thick] (1.45,0.00) -- (-0.00, -1.45);
\draw[blue, thick] (1.45,0.00) -- (1.03, -1.03);
\draw[blue, thick] (1.03,1.03) -- (-1.45, 0.00);
\draw[blue, thick] (1.03,1.03) -- (-1.03, -1.03);
\draw[blue, thick] (1.03,1.03) -- (-0.00, -1.45);
\draw[blue, thick] (1.03,1.03) -- (1.03, -1.03);
\draw[blue, thick] (0.00,1.45) -- (-1.03, 1.03);
\draw[blue, thick] (0.00,1.45) -- (-0.00, -1.45);
\draw[blue, thick] (0.00,1.45) -- (1.03, -1.03);
\draw[blue, thick] (-1.03,1.03) -- (-1.45, 0.00);
\draw[blue, thick] (-1.03,1.03) -- (-1.03, -1.03);
\draw[blue, thick] (-1.03,1.03) -- (-0.00, -1.45);
\draw[blue, thick] (-1.45,0.00) -- (-1.03, -1.03);
\draw[blue, thick] (-1.03,-1.03) -- (1.03, -1.03);
\draw[red, thick] (-6.55,0.00) -- (-6.74, 0.72);
\draw[red, thick] (-6.55,0.00) -- (-7.28, 1.26);
\draw[red, thick] (-6.55,0.00) -- (-8.00, 1.45);
\draw[red, thick] (-6.55,0.00) -- (-8.72, 1.26);
\draw[red, thick] (-6.55,0.00) -- (-9.26, 0.73);
\draw[red, thick] (-6.55,0.00) -- (-9.45, 0.00);
\draw[red, thick] (-6.55,0.00) -- (-9.26, -0.72);
\draw[red, thick] (-6.55,0.00) -- (-8.73, -1.26);
\draw[red, thick] (-6.55,0.00) -- (-8.00, -1.45);
\draw[red, thick] (-6.55,0.00) -- (-7.28, -1.26);
\draw[red, thick] (-6.55,0.00) -- (-6.74, -0.73);
\draw[red, thick] (-6.74,0.72) -- (-7.28, 1.26);
\draw[red, thick] (-6.74,0.72) -- (-8.00, 1.45);
\draw[red, thick] (-6.74,0.72) -- (-8.72, 1.26);
\draw[red, thick] (-6.74,0.72) -- (-9.26, 0.73);
\draw[red, thick] (-6.74,0.72) -- (-9.45, 0.00);
\draw[red, thick] (-6.74,0.72) -- (-9.26, -0.72);
\draw[red, thick] (-6.74,0.72) -- (-8.73, -1.26);
\draw[red, thick] (-6.74,0.72) -- (-8.00, -1.45);
\draw[red, thick] (-6.74,0.72) -- (-7.28, -1.26);
\draw[red, thick] (-6.74,0.72) -- (-6.74, -0.73);
\draw[red, thick] (-7.28,1.26) -- (-8.00, 1.45);
\draw[red, thick] (-7.28,1.26) -- (-8.72, 1.26);
\draw[red, thick] (-7.28,1.26) -- (-9.26, 0.73);
\draw[red, thick] (-7.28,1.26) -- (-9.45, 0.00);
\draw[red, thick] (-7.28,1.26) -- (-9.26, -0.72);
\draw[red, thick] (-7.28,1.26) -- (-8.73, -1.26);
\draw[red, thick] (-7.28,1.26) -- (-8.00, -1.45);
\draw[red, thick] (-7.28,1.26) -- (-7.28, -1.26);
\draw[red, thick] (-7.28,1.26) -- (-6.74, -0.73);
\draw[red, thick] (-8.00,1.45) -- (-8.72, 1.26);
\draw[red, thick] (-8.00,1.45) -- (-9.26, 0.73);
\draw[red, thick] (-8.00,1.45) -- (-9.45, 0.00);
\draw[red, thick] (-8.00,1.45) -- (-9.26, -0.72);
\draw[red, thick] (-8.00,1.45) -- (-8.73, -1.26);
\draw[red, thick] (-8.00,1.45) -- (-8.00, -1.45);
\draw[red, thick] (-8.00,1.45) -- (-7.28, -1.26);
\draw[red, thick] (-8.00,1.45) -- (-6.74, -0.73);
\draw[red, thick] (-8.72,1.26) -- (-9.26, 0.73);
\draw[red, thick] (-8.72,1.26) -- (-9.45, 0.00);
\draw[red, thick] (-8.72,1.26) -- (-9.26, -0.72);
\draw[red, thick] (-8.72,1.26) -- (-8.73, -1.26);
\draw[red, thick] (-8.72,1.26) -- (-8.00, -1.45);
\draw[red, thick] (-8.72,1.26) -- (-7.28, -1.26);
\draw[red, thick] (-8.72,1.26) -- (-6.74, -0.73);
\draw[red, thick] (-9.26,0.73) -- (-9.45, 0.00);
\draw[red, thick] (-9.26,0.73) -- (-9.26, -0.72);
\draw[red, thick] (-9.26,0.73) -- (-8.73, -1.26);
\draw[red, thick] (-9.26,0.73) -- (-8.00, -1.45);
\draw[red, thick] (-9.26,0.73) -- (-7.28, -1.26);
\draw[red, thick] (-9.26,0.73) -- (-6.74, -0.73);
\draw[red, thick] (-9.45,0.00) -- (-9.26, -0.72);
\draw[red, thick] (-9.45,0.00) -- (-8.73, -1.26);
\draw[red, thick] (-9.45,0.00) -- (-8.00, -1.45);
\draw[red, thick] (-9.45,0.00) -- (-7.28, -1.26);
\draw[red, thick] (-9.45,0.00) -- (-6.74, -0.73);
\draw[red, thick] (-9.26,-0.72) -- (-8.73, -1.26);
\draw[red, thick] (-9.26,-0.72) -- (-8.00, -1.45);
\draw[red, thick] (-9.26,-0.72) -- (-7.28, -1.26);
\draw[red, thick] (-9.26,-0.72) -- (-6.74, -0.73);
\draw[red, thick] (-8.73,-1.26) -- (-8.00, -1.45);
\draw[red, thick] (-8.73,-1.26) -- (-7.28, -1.26);
\draw[red, thick] (-8.73,-1.26) -- (-6.74, -0.73);
\draw[red, thick] (-8.00,-1.45) -- (-7.28, -1.26);
\draw[red, thick] (-8.00,-1.45) -- (-6.74, -0.73);
\draw[red, thick] (-7.28,-1.26) -- (-6.74, -0.73);
\filldraw[blue] (1.45,0.00) circle (2 pt);
\filldraw[blue] (1.03,1.03) circle (2 pt);
\filldraw[blue] (0.00,1.45) circle (2 pt);
\filldraw[blue] (-1.03,1.03) circle (2 pt);
\filldraw[blue] (-1.45,0.00) circle (2 pt);
\filldraw[blue] (-1.03,-1.03) circle (2 pt);
\filldraw[blue] (-0.00,-1.45) circle (2 pt);
\filldraw[blue] (1.03,-1.03) circle (2 pt);
\filldraw[red] (-6.55,0.00) circle (2 pt);
\filldraw[red] (-6.74,0.72) circle (2 pt);
\filldraw[red] (-7.28,1.26) circle (2 pt);
\filldraw[red] (-8.00,1.45) circle (2 pt);
\filldraw[red] (-8.72,1.26) circle (2 pt);
\filldraw[red] (-9.26,0.73) circle (2 pt);
\filldraw[red] (-9.45,0.00) circle (2 pt);
\filldraw[red] (-9.26,-0.72) circle (2 pt);
\filldraw[red] (-8.73,-1.26) circle (2 pt);
\filldraw[red] (-8.00,-1.45) circle (2 pt);
\filldraw[red] (-7.28,-1.26) circle (2 pt);
\filldraw[red] (-6.74,-0.73) circle (2 pt);

\draw[line width=0.5mm, blue] (0.0,0.0) ellipse (1.7 cm and 1.7 cm);
\node[blue] () at (2, -1.0) {$V$};
\node[blue] () at (2, -1.5) {$|V|=n$};

\draw[line width=0.5mm, red] (-8.0,0.0) ellipse (1.7 cm and 1.7 cm);
\node[red] () at (-10, -1.0) {$V_E$};
\node[red] () at (-10, -1.5) {$|V_E|=m$};

\draw[-{Straight Barb[right]}, line width=0.7mm, black] (-6.3,-0.1) -- (-1.7,-0.1);
\draw[-{Straight Barb[right]}, line width=0.7mm, black] (-1.7,0.1) -- (-6.3,0.1);
\node[rotate=0] at (-4.0,0.4) {$m-3$};
\node[rotate=0] at (-4.0,-0.4) {$n-2$};

\draw[line width = 0.7mm,->] (1.6,-0.58) to [out=-20,in=20,looseness=5] (1.6,0.58);
\node[rotate=-90] () at (3.5,0) {$n-4$}; 
\end{tikzpicture}
\end{center}

\caption{The graph $G'$.}
\label{fig:delta_compl_2}
\end{figure}

\paragraph{Parameter $\Delta(\overline{G'})$:} Hereafter we justify that $\Delta(\overline{G'}) = 6$.
\begin{itemize}
    \item Since $G$ is cubic, every vertex $u\in V_G$ is adjacent to exactly three other vertices, say $v,w,x$. Let $u',v',w',x'$ be the corresponding vertices in $G'$. Consequently, $u'$ is adjacent to all vertices in $G'$ except $v',w',x',e_{uv},e_{uw}$ and $e_{ux}$. Thus, $d_{\overline{G'}}(u) = 6$.

    \item Every $e_{uv} \in V_E$ is adjacent to all vertices in $G'$ except $u'$ and $v'$, where $u',v'$ are the vertices corresponding to $u,v$ in $G$. Thus, $d_{\overline{G'}}(e_{uv}) = 2$.
\end{itemize} 

\paragraph{Soundness:} We show that a graph $G$ has an independent set of size $k$ if and only if $G' = (V',E')$, constructed as explained above, contains a \PDS{} of size $k'=m+k$.

\begin{itemize}
    \item[$\implies$] First, let $I$ be an independent set of size $k$ in $G$. Let $I'$ be the set of vertices in $G'$ corresponding to the set $I$. Let $S = I' \cup V_E$. First, notice that $|S| = |I'|+|V_E| = m+k = k'$. Now, we have to show that all the vertices are satisfied with respect to $S$.
    \begin{itemize}
        \item Let $u \in I$ and $v,w,x$ be the neighbors of $u$ in $G$. Since $I$ is an independent set and $u \in I$, we know that $v,w,x \notin I$. Let $u',v',w'$ and $x'$ be the corresponding vertices in $G'$.
        Then, $u'$ has exactly three non-neighbors in $S$, namely $e_{uv},e_{uw},e_{ux}$, and three non-neighbors in $V'\setminus S$, namely $v',w',x'$. We recall that $|S|=k' > |V'|/2$ and thus, $|S| > |V'\setminus S|$. Then, $u'$ is satisfied with respect to $S$ if and only if 
        $$ \frac{d_{S}(u')}{|S|-1} = \frac{|S|-4}{|S|-1} \ge \frac{|\overline{S}|-3}{|\overline{S}|} = \frac{d_{\overline{S}}(u')}{|\overline{S}|}. $$
        By \cref{lemma:xgeqY}, the inequality above holds. We conclude that vertices $V_G$ in $S$ are satisfied.
        
        \item Let $e_{uv} \in V_E$. Then, $e_{uv}$ has two non-neighbors in $V'$, namely $u'$ and $v'$. Since $uv \in E$, and $I$ is an independent set, we know that at most one of them is in $S = V_E \cup I'$ and at least one of them is in $V'\setminus S$. Suppose for a contradiction that $e_{uv}$ is not satisfied with respect to $S$. Then
        $$ \frac{d_{S}(e_{uv})}{|S|-1} = \frac{|S|-2}{|S|-1} < \frac{|\overline{S}|-1}{|\overline{S}|} \geq \frac{d_{\overline{S'}}(e_{uv})}{|\overline{S'}|}\,, $$
        a contradiction. Hence, $e_{uv}$ is satisfied.
    \end{itemize}
    Thus, $S$ is a \PDS{} of size $k'$ in $G'$.

    \item[$\impliedby$] Let $S$ is a \PDS{} of size at least $k'$ in $G'$. Notice that for all $uv \in E$ it holds that $N_{G'}(u) \subseteq N_{G'}(e_{uv})$. Thus, for all $uv \in E$ such that $u \in S$ and $e_{uv} \notin S$, we have that $S' = (S\setminus \{u\}) \cup \{e_{uv}\}$ is a \PDS{} of size at least $k'$ in $G'$. Below, we show that the vertices remain satisfied with respect to $S'$.
    \begin{itemize}
        \item Let $x \in S \setminus \{u\}$. Since $S$ is a \PDS{} by assumption, we know that $x$ is satisfied with respect to $S$. Notice that $|S| = |S'|$, $N_{G'}(u) \subseteq N_{G'}(e_{uv})$ and $d_{S'}(x) \geq d_{S}(x)$. Hence, $x$ is satisfied with respect to $S'$.

        \item Also, $e_{uv}$ has at least one non-neighbor in $V'\setminus S'$, namely $u$ and at most one non-neighbor in $S'$ (it has two non-neighbors in total). Thus, 
        $$ \frac{d_{S}(e_{uv})}{|S'|-1} = \frac{|S'|-2}{|S'|-1} = 1 - \frac{1}{|S'|} > 1-\frac{1}{|V'\setminus S'|} = \frac{|V'\setminus S'|-1}{|V'\setminus S'|} = \frac{d_{V'\setminus S'}(e_{uv})}{|V'\setminus S'|}. $$
        Hence, $e_{uv}$ is satisfied with respect to $S'$.
    \end{itemize}
    By iterating the process of swapping $u'$ with $e_{vu}$, as explained above, we obtain a \PDS{} $S^*$ in $G'$ of size $\geq k'$ such that for all $uv\in E$, it holds that $u' \in S^*$ only if $e_{uv} \in S^*$. For the sake of simplicity, we refer to $S^*$ as $S$.

    Let $I$ be the set of vertices in $G$ corresponding to the set $I'=S \cap V_G$ in $G'$. Further, we show that $I$ is an independent set of $G$. Indeed, for $uv \in E$ if either $u' \in S$ or $v' \in S$ then $e_{uv} \in S$. Also, if $e_{uv} \in S$ then both $u'$ and $v'$ cannot be in $S$ otherwise $e_{uv}$ have two non-neighbors in $S$ and no non-neighbor in $V'\setminus S$, i.e.
    $$ \frac{d_{S}(e_{uv})}{|S|-1} = \frac{|S|-3}{|S|-1} < \frac{|V'\setminus S|}{|V'\setminus S|} = \frac{d_{V'\setminus S}(e_{uv})}{|V'\setminus S|}. $$
    Hence, for $u,v \in I$, it holds that $uv \notin E$ and $I$ is an independent set of $G$. Moreover, since $S = I' \cup (S \cap V_E)$, $|S| \geq k'=m+k$ and $|S \cap V_E| \leq |V_E| = m$, it holds that $|I'|= |I|\geq k$.
\end{itemize}
This concludes the proof of \cref{thm:MaxPDS_NPh_Delta_complement}.
\end{proof}

Notice, that the \PDS{} considered in the proof of \cref{thm:MaxPDS_NPh_Delta_complement} induces a connected subgraph of $G'$. This implies the following result.

\begin{corollary}
    \connmaxPDS{} is \NP{}-hard on graphs $G$ such that $\Delta(\overline{G}) = 6$.
\end{corollary}

Note that we only use the fact that $G$ is cubic to obtain $\Delta(\overline{G'})=6$. Thus, the reduction works for any graph $G$.
Also, the \textsc{Independent Set} is \NP{}-hard on planar graphs (see \cite{indep_cubic}). If $G$ is a planar graph, then the complement of the graph $G'$ constructed as explained above is planar as well. Thus, using the exact same reduction, we can show that \maxPDS{} is \NP{}-hard when restricted to the complement of planar graphs.

\begin{corollary}
\maxPDS{} is \NP{}-hard when restricted to the complement of planar graphs.
\end{corollary}

Further, we present \NP{}-harness results on graphs which are complements of $2$-degenerate bipartite graphs. Before we move on to the main theorem, we introduce some preliminary results.  

\begin{lemma}\label{rk:tech_degen_compl}
Let $n,k,m$ be positive integers such that $n$ is even, $m=\frac{3}{2}n$ and $2 \leq k \leq n$. Let $x = n-k+3$. Then, 
$$\frac{xm+k-1-3x}{xm+k-1} \geq \frac{n-k}{n-k+2}.$$
\end{lemma}

\begin{proof}
Suppose for a contradiction that $\frac{xm+k-1-3x}{xm+k-1} < \frac{n-k}{n-k+2}$.
\begin{align*}
& \frac{xm+k-1-3x}{xm+k-1} < \frac{n-k}{n-k+2} \\
\iff \quad & (xm+k-1-3x)(n-k+2) < (xm+k-1)(n-k) \\
\iff \quad & (xm+k-1)(n-k+2) < (xm+k-1)(n-k) + 3x(n-k+2) \\
\iff \quad & 2xm+2k-2 < 3x(n-k+2) \\
\iff \quad & 3xn + 2k-2 < 3xn-3xk+6x \\
\iff \quad & 3xk + 2k-2 < 6x,
\end{align*}
which contradicts the fact that $k \geq 2$.
\end{proof}

Hereafter, we move on to the main theorem.

\begin{theorem}\label{thm:maxPDS_NPh_degen_complement}
\maxPDS{} is \NP{}-hard when restricted to graphs $G$ such that $\degen(\overline{G})=2$ and $\overline{G}$ is bipartite.
\end{theorem}

\begin{proof}
Similarly, as above, we provide a polynomial-time reduction from the \textsc{Independent Set} on cubic graphs.

\paragraph{Construction:} Given an instance of the \textsc{Independent Set}, that is, a cubic graph $G = (V,E)$ and a positive integer $k$, we construct $G' =(V',E')$, a positive integer $k'$, an instance of \maxPDS{} as follows. Let $|V| =  n$, $ |E| = m$, $x=n-k+3$, and $k \geq 4$.

\begin{itemize}
    \item The vertex set $V'$ is constructed as follows.
    \begin{itemize}
        \item We create a copy of vertices in $V$. We denote by $V_G$ the set formed by such vertices.

        \item We create $x$ vertices $e_{uv}^{1},...,e_{uv}^{x}$ for each $uv \in E$. We denote the corresponding set of vertices by $V_E$.

        \item We create two new vertices $a^*,b^*$.
    \end{itemize}

    \item The edge set $E'$ is constructed as follows.
    \begin{itemize}
        \item $V_G$ forms a clique in $G'$,

        \item $V_E \cup \{a^*,b^*\}$ forms a clique in $G'$.

        \item For every $uv \in E$ and every $i \in \intInterval{1,x}$, $e_{uv}^{i}$ is adjacent to all vertices in $V_G$ other than $u'$ and $v'$, where $u'$ and $v'$ are vertices of $G'$ corresponding to the vertices $u,v$ in $G$.
    \end{itemize}
\end{itemize}

Note that intuitively, the complement of $G'$ is a graph obtained by subdividing each edge $uv \in E$ into a path of length 2. Then, we have $x$ copies of the vertices $e_{uv}$. Finally, two vertices $a^*,b^*$ are added along with all possible edges between them and $V_E$. See \Cref{fig:degen_complement} for an illustration of $G'$.

\begin{figure}[!ht]
\centering
\begin{tikzpicture}[scale=1]
\draw[thick,blue] (7.30,0.00) -- (6.92,0.92);
\draw[thick,blue] (7.30,0.00) -- (6.00,1.30);
\draw[thick,blue] (7.30,0.00) -- (5.08,0.92);
\draw[thick,blue] (7.30,0.00) -- (4.70,0.00);
\draw[thick,blue] (7.30,0.00) -- (5.08,-0.92);
\draw[thick,blue] (7.30,0.00) -- (6.00,-1.30);
\draw[thick,blue] (7.30,0.00) -- (6.92,-0.92);
\draw[thick,blue] (6.92,0.92) -- (6.00,1.30);
\draw[thick,blue] (6.92,0.92) -- (5.08,0.92);
\draw[thick,blue] (6.92,0.92) -- (4.70,0.00);
\draw[thick,blue] (6.92,0.92) -- (5.08,-0.92);
\draw[thick,blue] (6.92,0.92) -- (6.00,-1.30);
\draw[thick,blue] (6.92,0.92) -- (6.92,-0.92);
\draw[thick,blue] (6.00,1.30) -- (5.08,0.92);
\draw[thick,blue] (6.00,1.30) -- (4.70,0.00);
\draw[thick,blue] (6.00,1.30) -- (5.08,-0.92);
\draw[thick,blue] (6.00,1.30) -- (6.00,-1.30);
\draw[thick,blue] (6.00,1.30) -- (6.92,-0.92);
\draw[thick,blue] (5.08,0.92) -- (4.70,0.00);
\draw[thick,blue] (5.08,0.92) -- (5.08,-0.92);
\draw[thick,blue] (5.08,0.92) -- (6.00,-1.30);
\draw[thick,blue] (5.08,0.92) -- (6.92,-0.92);
\draw[thick,blue] (4.70,0.00) -- (5.08,-0.92);
\draw[thick,blue] (4.70,0.00) -- (6.00,-1.30);
\draw[thick,blue] (4.70,0.00) -- (6.92,-0.92);
\draw[thick,blue] (5.08,-0.92) -- (6.00,-1.30);
\draw[thick,blue] (5.08,-0.92) -- (6.92,-0.92);
\draw[thick,blue] (6.00,-1.30) -- (6.92,-0.92);
\draw[thick,red] (1.30,0.00) -- (1.13,0.65);
\draw[thick,red] (1.30,0.00) -- (0.65,1.13);
\draw[thick,red] (1.30,0.00) -- (0.00,1.30);
\draw[thick,red] (1.30,0.00) -- (-0.65,1.13);
\draw[thick,red] (1.30,0.00) -- (-1.13,0.65);
\draw[thick,red] (1.30,0.00) -- (-1.30,0.00);
\draw[thick,red] (1.30,0.00) -- (-1.13,-0.65);
\draw[thick,red] (1.30,0.00) -- (-0.65,-1.13);
\draw[thick,red] (1.30,0.00) -- (-0.00,-1.30);
\draw[thick,red] (1.30,0.00) -- (0.65,-1.13);
\draw[thick,red] (1.30,0.00) -- (1.13,-0.65);
\draw[thick,red] (1.13,0.65) -- (0.65,1.13);
\draw[thick,red] (1.13,0.65) -- (0.00,1.30);
\draw[thick,red] (1.13,0.65) -- (-0.65,1.13);
\draw[thick,red] (1.13,0.65) -- (-1.13,0.65);
\draw[thick,red] (1.13,0.65) -- (-1.30,0.00);
\draw[thick,red] (1.13,0.65) -- (-1.13,-0.65);
\draw[thick,red] (1.13,0.65) -- (-0.65,-1.13);
\draw[thick,red] (1.13,0.65) -- (-0.00,-1.30);
\draw[thick,red] (1.13,0.65) -- (0.65,-1.13);
\draw[thick,red] (1.13,0.65) -- (1.13,-0.65);
\draw[thick,red] (0.65,1.13) -- (0.00,1.30);
\draw[thick,red] (0.65,1.13) -- (-0.65,1.13);
\draw[thick,red] (0.65,1.13) -- (-1.13,0.65);
\draw[thick,red] (0.65,1.13) -- (-1.30,0.00);
\draw[thick,red] (0.65,1.13) -- (-1.13,-0.65);
\draw[thick,red] (0.65,1.13) -- (-0.65,-1.13);
\draw[thick,red] (0.65,1.13) -- (-0.00,-1.30);
\draw[thick,red] (0.65,1.13) -- (0.65,-1.13);
\draw[thick,red] (0.65,1.13) -- (1.13,-0.65);
\draw[thick,red] (0.00,1.30) -- (-0.65,1.13);
\draw[thick,red] (0.00,1.30) -- (-1.13,0.65);
\draw[thick,red] (0.00,1.30) -- (-1.30,0.00);
\draw[thick,red] (0.00,1.30) -- (-1.13,-0.65);
\draw[thick,red] (0.00,1.30) -- (-0.65,-1.13);
\draw[thick,red] (0.00,1.30) -- (-0.00,-1.30);
\draw[thick,red] (0.00,1.30) -- (0.65,-1.13);
\draw[thick,red] (0.00,1.30) -- (1.13,-0.65);
\draw[thick,red] (-0.65,1.13) -- (-1.13,0.65);
\draw[thick,red] (-0.65,1.13) -- (-1.30,0.00);
\draw[thick,red] (-0.65,1.13) -- (-1.13,-0.65);
\draw[thick,red] (-0.65,1.13) -- (-0.65,-1.13);
\draw[thick,red] (-0.65,1.13) -- (-0.00,-1.30);
\draw[thick,red] (-0.65,1.13) -- (0.65,-1.13);
\draw[thick,red] (-0.65,1.13) -- (1.13,-0.65);
\draw[thick,red] (-1.13,0.65) -- (-1.30,0.00);
\draw[thick,red] (-1.13,0.65) -- (-1.13,-0.65);
\draw[thick,red] (-1.13,0.65) -- (-0.65,-1.13);
\draw[thick,red] (-1.13,0.65) -- (-0.00,-1.30);
\draw[thick,red] (-1.13,0.65) -- (0.65,-1.13);
\draw[thick,red] (-1.13,0.65) -- (1.13,-0.65);
\draw[thick,red] (-1.30,0.00) -- (-1.13,-0.65);
\draw[thick,red] (-1.30,0.00) -- (-0.65,-1.13);
\draw[thick,red] (-1.30,0.00) -- (-0.00,-1.30);
\draw[thick,red] (-1.30,0.00) -- (0.65,-1.13);
\draw[thick,red] (-1.30,0.00) -- (1.13,-0.65);
\draw[thick,red] (-1.13,-0.65) -- (-0.65,-1.13);
\draw[thick,red] (-1.13,-0.65) -- (-0.00,-1.30);
\draw[thick,red] (-1.13,-0.65) -- (0.65,-1.13);
\draw[thick,red] (-1.13,-0.65) -- (1.13,-0.65);
\draw[thick,red] (-0.65,-1.13) -- (-0.00,-1.30);
\draw[thick,red] (-0.65,-1.13) -- (0.65,-1.13);
\draw[thick,red] (-0.65,-1.13) -- (1.13,-0.65);
\draw[thick,red] (-0.00,-1.30) -- (0.65,-1.13);
\draw[thick,red] (-0.00,-1.30) -- (1.13,-0.65);
\draw[thick,red] (0.65,-1.13) -- (1.13,-0.65);
\draw[thick,vert] (-4.7,1.1) -- (-4.7,-1.1);

\filldraw[blue] (7.30,0.00) circle (2pt);
\filldraw[blue] (6.92,0.92) circle (2pt);
\filldraw[blue] (6.00,1.30) circle (2pt);
\filldraw[blue] (5.08,0.92) circle (2pt);
\filldraw[blue] (4.70,0.00) circle (2pt);
\filldraw[blue] (5.08,-0.92) circle (2pt);
\filldraw[blue] (6.00,-1.30) circle (2pt);
\filldraw[blue] (6.92,-0.92) circle (2pt);
\filldraw[red] (1.30,0.00) circle (2pt);
\filldraw[red] (1.13,0.65) circle (2pt);
\filldraw[red] (0.65,1.13) circle (2pt);
\filldraw[red] (0.00,1.30) circle (2pt);
\filldraw[red] (-0.65,1.13) circle (2pt);
\filldraw[red] (-1.13,0.65) circle (2pt);
\filldraw[red] (-1.30,0.00) circle (2pt);
\filldraw[red] (-1.13,-0.65) circle (2pt);
\filldraw[red] (-0.65,-1.13) circle (2pt);
\filldraw[red] (-0.00,-1.30) circle (2pt);
\filldraw[red] (0.65,-1.13) circle (2pt);
\filldraw[red] (1.13,-0.65) circle (2pt);
\filldraw[vert] (-4.7,1.1) circle (2pt);
\filldraw[vert] (-4.7,-1.1) circle (2pt);

\node[vert] () at (-5,1.3) {$a^*$};
\node[vert] () at (-5,-1.3) {$b^*$};

\draw[black, line width=0.3mm] (-1.3,0.7) -- (-4.6,-1.1);
\draw[black, line width=0.3mm] (-1.3,0.7) -- (-4.6,1.1);
\draw[black, line width=0.3mm] (-1.5,0.0) -- (-4.6,-1.1);
\draw[black, line width=0.3mm] (-1.5,0.0) -- (-4.6,1.1);
\draw[black, line width=0.3mm] (-1.3,-0.7) -- (-4.6,-1.1);
\draw[black, line width=0.3mm] (-1.3,-0.7) -- (-4.6,1.1);
\draw[blue, line width=0.5mm] (6.0,0.0) ellipse (1.5 cm and 1.5 cm);
\node[blue] () at (7.6, 1.9) {$V$};
\node[blue] () at (7.9, 1.4) {$|V|=n$};
\draw[red, line width=0.5mm] (0.0,0.0) ellipse (1.5 cm and 1.5 cm);
\node[red] () at (1.3, 2.1) {$V_E$};
\node[red] () at (1.7, 1.6) {$|V_E|=xm$};
\draw[-{Straight Barb[right]}, black, line width=0.7mm] (1.5,-0.1) -- (4.5,-0.1);
\draw[-{Straight Barb[right]}, black, line width=0.7mm] (4.5,0.1) -- (1.5,0.1);
\node[rotate=0] at (3.0,0.4) {$x(m-3)$};
\node[rotate=0] at (3.0,-0.4) {$n-2$};
\end{tikzpicture}
\caption{The graph $G'$. Note that the green vertices correspond to the vertices $a^*$ and $b^*$.}
\label{fig:degen_complement}
\end{figure}
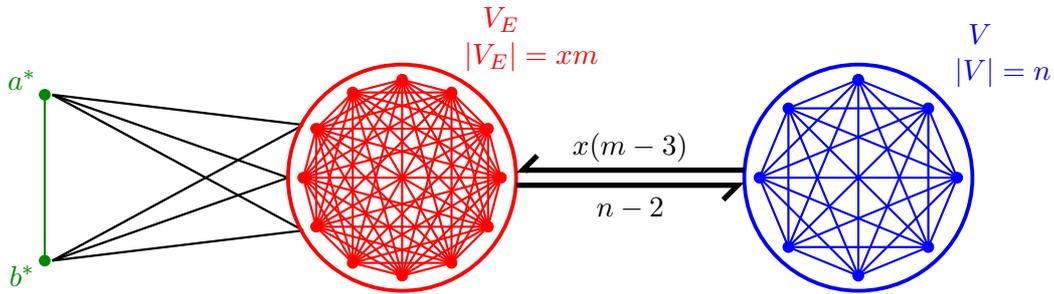

Finally, let $k'=xm+k$. Note that $|V'| = xm+n+2$, $x \geq 2$ and $m = \frac{3}{2}n$. Thus, $xm > \frac{|V'|}{2}$.

\paragraph{Parameter $\degen$ and bipartition of $\overline{G'}$:} Hereafter we justify that $\degen(\overline{G'})=2$ and that $\overline{G'}$ is bipartite.
\begin{itemize}

    \item \underline{Degeneracy:} We provide an elimination order of vertices of $\overline{G'}$ that only eliminates vertices of degree at most $2$ at each step. Notice that all vertices in $V_E$ have degree $2$ in $\overline{G'}$. We start by removing these vertices. Then, in $\overline{G'}[V_G \cup \{a^*,b^*\}]$, the vertices of $V_G$ have degree $2$. More specifically, their neighbors are $a^*$ and $b^*$. Thus, we proceed by removing the vertices of $V_G$. Finally, vertices $a^*$ and $b^*$, are of degree $0$. Thus, $\degen(\overline{G'})=2$.

    \item \underline{Bipartition:} Note that $G'$ can be partitioned into two cliques: $V_G$ and $V_{e} \cup \{a^*,b^*\}$. Thus $\overline{G'}$ can be partitioned into two independent sets and $\overline{G'}$ is bipartite.
\end{itemize}

\paragraph{Soundness:} We show that $G$ contains an independent set of size $k$ if and only if $G'$ contains a \PDS{} of size at least $k'=xm+k$.

\begin{itemize}
    \item[$\implies$] Let $I$ be an independent set of size $k$ in $G$. Let $I'$ be the set of vertices in $G'$ corresponding to the set $I$. Let $S=I' \cup V_E$. It holds that $|S|=xm+k$. Then, we show that each vertex is satisfied with respect to $S$.
    \begin{itemize}
        \item Let $u \in I'$. Notice that $u$ has $3x$ non-neighbors in $S$, namely $e_{uv}^{i}$, for every $v$ such that $uv \in E$ and every $i \in \intInterval{1,x}$. Also, $u$ has 2 non-neighbor in $\overline{S}$, namely $a^*,b^*$. Then, $u$ is satisfied with respect to $S$ if and only if the following equation is satisfied:
        $$ \frac{d_{S}(u)}{|S|-1} = \frac{|S|-1-3x}{|S|-1} \geq \frac{|\overline{S}|-2}{|\overline{S}|} = \frac{d_{\overline{S}}(u)}{|\overline{S}|}. $$
        \Cref{rk:tech_degen_compl} implies that the equation above holds. Thus, $u$ is satisfied with respect to $S$.

        \item Let $e_{uv}^{i} \in V_E$ for some $uv \in E$ and some $i \in [\![1,x]\!]$. Since $uv \in E$, we know that at most one of the vertices $u,v$ is in $I$. Thus, $e_{uv}^i$ has at most one non-neighbor in $S$ and at least one non-neighbor in $\overline{S}$. We recall that $|S| > |\overline{S}|$. Then, we have 
        $$ \frac{d_{S}(e_{uv}^i)}{|S|-1} \geq \frac{|S|-2}{|S|-1} \geq \frac{|\overline{S}|-1}{|\overline{S}|} \geq \frac{d_{\overline{S}}(e_{uv}^i)}{|\overline{S}|}, $$
        and $e_{uv}^i$ is satisfied with respect to $S$.
    \end{itemize}
    Thus, $S$ is a \PDS{} of size $k'$ in $G'$.

    \item[$\impliedby$] Let $S$ be a \PDS{} of size at least $k'$ in $G'$. Recall that for every $u,v \in E$, there are $x=n-k+3$ copies of $e_{uv}$ in $G'$. Also, $|\overline{S}| \leq n-k+2$, thus $S$ contains at least one copy of $e_{uv}$. Finally, $e_{uv}^{i} \in S$ is satisfied if and only if $S$ contains at most one of the vertices $u$ or $v$. 
    
     Further, let us consider the size of $S$. Since $S$ contains at least one copy of $e_{uv}$, for every $uv \in E$, we have that the set of vertices $I$ in $G$ corresponding to the set $I' = S \cap V_G$ is an independent set of $G$. Finally, $S$ contains at least $k'=xm+k$ vertices, thus, $|I|= |I'| \geq = k+(xm-|S\cap V_E|)-2$. 
    Suppose for a contradiction that $|I| < k$. Then $|S\cap V_E| = xm-1$ and $a^*,b^* \in S$. Let $e_{uv}^{i}$ be the vertex in $V_E$ which is not contained in $S$. Since $I$ contains at least $k-1 \geq 3$ vertices, it contains at least one vertex which is not equal to $u$ or $v$. Let $w \in I$ be this vertex. It holds that all non-neighbors of $w$, namely $e_{w\_}^i$ and $a^*,b^*$, are contained in $S$. Then, $w \in S$ is not satisfied with respect to $S$, a contradiction. Thus $|I| \geq k$ and $I$ is an independent set of size at least $k$ in $G$. \qedhere
\end{itemize}
\end{proof}

Similarly as before, the \PDS{} considered in the proof of \cref{thm:maxPDS_NPh_degen_complement} induces a connected subgraph of $G'$. This implies the following result.

\begin{corollary}
    \connmaxPDS{} is \NP{}-hard on graphs $G$ such that $\degen(\overline{G}) = 2$ and $\overline{G}$ is bipartite.
\end{corollary}

A \emph{clique cover} of a given graph is a partition of the vertices into cliques. The minimum $k$ for which a clique cover exists is called the \emph{clique cover number}. Notice that if the complement $\overline{G}$ of a graph $G$ is bipartite, then $G$ has a clique cover number equal to $2$. This implies the following result.

\begin{corollary}
    \maxPDS{} and \connmaxPDS{} are para-\NP{}-hard parameterized by clique cover number.
\end{corollary}

\section{Polynomial-time solvable cases}
\label{sec:polynsolvablecases}

In this section, we describe polynomial-time algorithms for {\maxPDS} and \connmaxPDS{} on specific classes of graphs.
In \cref{subsec:h=2}, we present a polynomial-time algorithm for solving \maxPDS{} on graphs with $h \leq 2$.
Then, in \cref{subsec:graphsofdegreeatleastN-3}, we show how to find a (connected) \PDS{} of maximum size in graphs $G$ where $\Delta(\overline{G}) \leq 2$, where $\overline{G}$ represents the complement of $G$. This result implies that \maxPDS{} can be solved in polynomial time on graphs with $n$ vertices with minimum degree at least $(n-3)$ in polynomial time, i.e.\ graphs $G$ with $\Delta(\overline{G}) \leq 2$.

\subsection{Graphs with $h(G)\le2$}
\label{subsec:h=2}
In this subsection, we show how to find a \PDS{} of maximum size in graphs with an $h\le 2$ in polynomial time. Notice that since $h\le \Delta$, the result implies polynomial time complexity for solving \maxPDS{} on a graph with $\Delta \leq 2$.

A \emph{matching} in a graph $G$ is a set of pairwise non-adjacent edges. To avoid unnecessary complications and trivial cases in the proof of the main result of this section, we explain how to solve \maxPDS{} on graphs with $h\le 1$. First, in the lemma introduced below, we explain the structural properties of graphs with $h\le 1$. 

\begin{lemma}\label{lem:structure_h_1}
Let $G=(V,E)$ be a graph with $h(G) \leq 1$. Then the vertex set $V$ can be partitioned into sets $V_I, V_M$ and $V_*$, where $V_I$ is an independent set, the edges of $G[V_M]$ form a matching and $G[V_*]$ is a star.
\end{lemma}

\begin{proof}
By the definition of the $h-$index,  all vertices in $G$ have a degree of at most $1$, except for one, which we denote by $u$. Then, in $G[V \setminus \{u\}]$ all vertices are of degree at most $1$ and the subgraph consists of a matching $M$ and an independent set $I'$. Let $V_*' = N(u)$ in $G$ and let $V_* = V_*' \cup \{u\}$. Clearly, $V_*$ is a star in $G$, $M$ is a matching in $G$ and $I = I' \setminus V_*'$ is an independent set in $G$. Then, by definition there are no edges between $I$ and $V_*$, and all vertices of $M$ have one neighbor which is contained in $M$. Thus, are no edges between $M$ and $I$, nor between $M$ and $V_*$.
\end{proof}

In the following lemma, we show that there always exists a \PDS{} of size either $|V|-1$ or $|V|-2$ in a graph $G=(V, E)$ with $h = 1$.

\begin{lemma}\label{lem:h_1_poly}
Let $G=(V,E)$ be a graph with $h(G)=1$. Then, a PDS $S$ of maximum size can be computed in $\mathcal{O}(|V|)$. Moreover, $|V|-2\le |S| \le |V|$.
\end{lemma}

\begin{proof}
    By the \cref{lem:structure_h_1}, we know that $V = V_I \cup V_M \cup V_*$, where $V_I, V_M$ and $V_*$ represent an independent set, matching and a star, respectively. We distinguish three cases.
    \begin{itemize}
    \item If $V_I$ and $V_M$ are empty, then for any $v \in V_*$ with $d(v)=1$, $S = V(G) \setminus \{v\}$ is a maximum \PDS{} in $G$. Let $u\in V_*$, with $d(u) > 1$. Then, $u$ is the only vertex of $G$ with a neighbor in $\overline{S}$ and the following inequality holds for $u$:
    $$\frac{d(u)}{|V|-1} = 1 \geq \frac{1}{|\{v\}|}\,.$$ Hence, all vertices are satisfied with respect to $S$ and $S$ is a \PDS{} of size $\vert V\vert -1$.
    
    \item If $V_I = \emptyset$, let $uv$ be an edge of $V_M$. Then, we claim that $S=V \setminus \{u,v\}$ is a \PDS{} of maximum size in $G$. 
    Indeed, $u$ and $v$ have no neighbors in $S$ and hence every vertex in $S$ is satisfied. Moreover, for any subset $S' \subseteq V$ of size $|V|-1$, there is a vertex $u \in S'$ with one neighbor in $\overline{S'}$. Since in this case, $G$ contains no universal vertices, $u$ is not adjacent to one vertex in $S'$ but is adjacent to the vertex in $\overline{S}$ and thus, $u$ is not satisfied with respect to $S'$ and the statement follows.

    \item Finally, if $V_I\neq \emptyset$, let $u \in I$. We claim that $S = V \setminus \{u\}$ is a \PDS{} of maximum size in $G$. Indeed, since $u$ has no neighbors in $S$, all vertices of $S$ are satisfied.
\end{itemize}
Hence, the statement follows. Moreover, notice that the complexity of determining a \PDS{} of the maximum size is $\mathcal{O}(|V|)$. The complexity follows from the complexity of determining the sets $V_I, V_M$ and $V_*$. 
\end{proof}

By \cref{lem:h_1_poly}, we know that a \PDS{} of maximum size in graphs with $h=1$ can be computed in time $\mathcal{O}(n)$. Now, we move on to graphs with $h= 2$. To show that a \PDS{} of maximum size can be found in polynomial time in such graphs, we use 2-dimensional \textsc{Knapsack Problem}, or shortly \textsc{2d Knapsack Problem}. To this end, we will first a definition of a \textsc{Knapsack Problem}.

\noindent\rule[0.5ex]{\linewidth}{1pt}
{Knapsack Problem:}\\
\noindent\rule[0.5ex]{\linewidth}{0.5pt}
{\em Input:}\hspace{0.47cm} Values $(v_i)_{1\leq i \leq n}$ with weights $(w_i)_{1 \leq i \leq n}$, target value $T_V$ and weight constraint $W$.\\
{\em Output:}\hspace{0.2cm} A subset of indices $I \subseteq [\![1,n]\!]$ such that $\sum_{i \in I} v_i \geq T_V$ and $\sum_{i \in I} w_i \leq W$.\\
\noindent\rule[0.5ex]{\linewidth}{1pt}\\
Then, in the \textsc{2d Knapsack Problem}, the weights are two dimensional, i.e.\ $w_1,...,w_n, W \in \Z^2$ (with $(a,b) \leq (c,d) \iff a \leq c$ and $b \leq d$). Note that this problem is known to be weakly \NP{}-complete (see \cite{garey_johnson}). More specifically, in the general case, the complexity of the \textsc{2d Knapsack Problem} is $\mathcal{O}(n \cdot V \cdot W_1 \cdot W_2)$ where $(W_1,W_2)=W$. Below, we use the \textsc{exact 2d Knapsack Problem} where we search for a subset $I$ such that $\sum_{i \in I} v_i = T_V$ instead. This problem can be solved with the same time complexity.

Further, in the lemma below, we explain the structural properties of graphs with $h=2$. To this end, we introduce relevant notation. Let $G=(V,E)$ be a graph with $h(G)=2$. Let $u^*,v^* \in V$ be the vertices of maximum degree in $G$. Moreover, let $P_1,...,P_p$ and $C_1,...,C_{c}$ be the paths and cycles of $G[V \setminus \{u^*,v^*\}]$, respectively. Note, that we classify the isolated vertices as paths of length $0$.

\begin{lemma}\label{lemma:h_21}
Let $G=(V,E)$ be a graph with $h(G)=2$. We have that $G[V \setminus \{u^*,v^*\}]$ is a disjoint union of paths and cycles. Moreover, the vertices of cycles in $G[V \setminus \{u^*,v^*\}]$ are not adjacent to either $u^*$ or $v^*$. Let $P=u_1,...,u_{n_i}$ be a path in $G[V \setminus \{u^*,v^*\}]$, for some $n_i < n$.  Then, exactly one (or two, if $n_i=1$) of the following holds,
\begin{enumerate}
    \item[$i$.] vertices of $P$ are not adjacent to either $u^*$ or $v^*$;
    \item[$ii$.] $u_1u^* \in E$;
    \item[$iii$.] $u_1u^*,u_{n_i}u^* \in E$;
    \item[$iv$.] $u_1v^* \in E$;
    \item[$v$.] $u_1v^*,u_{n_i}v^* \in E$;
    \item[$vi$.]\label{vi} $u_1u^*, u_{n_i}v^* \in E$;
    \item[$vii$.]\label{vii} $u_{n_i}u^*, u_1v^* \in E$;
\end{enumerate}
\end{lemma}

\begin{proof}
Let $G' = G[V \setminus \{u^*,v^*\}]$. By definition of the $h$-index, we have $d(u) \leq 2$ for every $u \in V \setminus \{u^*,v^*\}$. Hence $\Delta(G') \leq 2$ and $G'$ is a disjoint union of paths and cycles. Let $C_i$ be a cycle in $G'$, for some $i\le c$, then the vertices of $C_i$ have degree $2$ in $G'$ and hence are not adjacent to either $u^*$ or $v^*$. For a path $P = u_1,...,u_{n_i}$, for ${n_i} < n$, we have that, for $2 \leq j \leq {n_i}-1$, $u_j$ has degree 2 in $G'$ and degree at most $2$ in $G$ and hence $u_j$ is not adjacent to either $u^*$ or $v^*$. Finally, if ${n_i}>1$, $u_1$ and $u_{n_i}$ have degree 1 in $G'$, thus they are adjacent to at most one of the vertices $\{u^*,v^*\}$. If ${n_i}=1$, then $P$ is an isolated vertex and $u_1=u_{n_i}$ could either be adjacent to $u^*$, $v^*$, both or neither of them.
\end{proof}

Let $S$ be a \PDS{} in $G$, where $G$ is a graph with $h= 2$. By the \cref{lemma:h_21}, we know that $G[V \setminus \{u^*,v^*\}]$ is a disjoint union of paths and cycles. In the lemma below we show that the paths and cycles in $G[V \setminus \{u^*,v^*\}]$ are either entirely contained in $S$ or $\overline{S}$.

\begin{lemma}\label{lemma:h_22}
Let $G=(V,E)$ be a graph with $h(G)=2$. Let $S$ be a \PDS{} in $G$ of size at least $\frac{|V|}{2}+1$. For every $1\leq i \leq p$ (resp.\ $1 \leq j \leq c$), either $P_i \subseteq S$ or $P_i \subseteq \overline{S}$ (resp.\ $C_j \subseteq S$ or $C_j \subseteq \overline{S}$). Moreover, exactly one of the following holds,
\begin{itemize}
    \item[$i.$] $u^*,v^* \notin S$
    \item[$ii.$] $u^* \in S$, $v^* \notin S$,
    \item[$iii.$] $u^* \notin S$, $v^* \in S$,
    \item[$iv.$] $u^*,v^* \in S$.
\end{itemize}
Finally, if $u^* \notin S$ (resp.\ $v^* \notin S$) and an endpoint of a path $P_i$ is adjacent to $u^*$ (resp.\ $v^*$), then $P_i \subseteq \overline{S}$.
\end{lemma}

\begin{proof}
The vertices of any path or cycle have degree at most 2. Hence, \cref{cor:one_to_all} implies that every path or cycle is entirely contained in $S$ or in $\overline{S}$. If $u^* \notin S$, then by the contraposition of \cref{cor:one_to_all} it follows that if one of the endpoints of a path $P_i$ is adjacent to $u^*$, then $P_i \subseteq \overline{S}$.
\end{proof}

Then, given a graph $G$ with $h=2$, we define an equivalent instance of the \textsc{exact 2d Knapsack Problem}.
\vspace{-0.3cm}
\paragraph{Construction:} Let $G=(V,E)$ be a graph with $h(G) = 2$ and $S^* \subseteq \{u^*,v^*\}$. Let $k$ be an integer such that $\frac{\vert V \vert}{2} \leq k \leq \vert V \vert$. Hereafter we show how to construct an equivalent instance of the \textsc{exact 2d Knapsack Problem}.

Let $\mathcal{P}$ be the union of all paths where at least one endpoint is adjacent to a vertex in $\{u^*,v^*\} \setminus S^*$. Let $P_1,..., P_p$ be the paths in $G[V \setminus \{u^*,v^*\}]$ that are not contained in $\mathcal{P}$. 

\begin{itemize}
    \item For each path $P_i$, we create an item with value $v_i = |P_i|$ and weight $w_i = (d_{P_i}(u^*), d_{P_i}(v^*))$.

    \item For each cycle $C_j$, we create an item with value $v_{p+j} = |C_j|$ and weight $w_{p+j} = (0,0)$.

    \item The target value is $T_V = |V|-k-|\mathcal{P}|$.

    \item The weight constraint is $W=(W_1, W_2)$, where:
    \begin{itemize}
        \item $W_1 = d(u^*)\frac{|V|-k}{|V|-1}$ and $W_2 = d(v^*)\frac{|V|-k}{|V|-1}$ if $u^*,v^* \in S^*$,
        \item $W_1 = (d(u^*)-1)\frac{|V|-k}{|V|-1}$ and $W_2 = \infty$ if $u^* \in S^*$ and $v^* \notin S^*$ and $u^*v^* \in E$,
        \item $W_1 = d(u^*)\frac{|V|-k}{|V|-1}$ and $W_2 = \infty$ if $u^* \in S^*$ and $v^* \notin S^*$ and $u^*v^* \notin E$,
        \item $W_1 = \infty$ and $W_2 = (d(v^*)-1)\frac{|V|-k}{|V|-1}$ if $u^* \notin S^*$ and $v^* \in S^*$ and $u^*v^* \in E$,
        \item $W_1 = \infty$ and $W_2 = d(v^*)\frac{|V|-k}{|V|-1}$ if $u^* \notin S^*$ and $v^* \in S^*$ and $u^*v^* \notin E$,
        \item $W_1 = \infty$ and $W_2 = \infty$ if $u^*,v^* \notin S^*$.
    \end{itemize}
\end{itemize}

Next, we prove the equivalence between a graph $G$ with $h=2$, an instance of \maxPDS{}, and the instance of the \textsc{exact 2d Knapsack Problem} constructed above.

\begin{lemma}
\label{lemma:h=2reduction}
Let $G=(V,E)$ be a graph with $h(G) = 2$. Let $S^* \subseteq \{u^*,v^*\}$. Let $k$ be an integer such that $\frac{\vert V \vert}{2} + 1 \leq k \leq \vert V \vert$. Let $\mathcal{I}$ be the instance of the \textsc{exact 2d Knapsack Problem} constructed above. Then there exists a \PDS{} $S$ in $G$ of size $k$ such that $S \cap \{u^*,v^*\} = S^*$ if and only if $\mathcal{I}$ is a yes-instance of the \textsc{exact 2d Knapsack Problem}.
\end{lemma}

\begin{proof}
\begin{itemize}
    \item[$\implies$] Let $S$ be a \PDS{} of size $k$ such that $S \cap \{u^*,v^*\} = S^*$. By the \cref{lemma:h_21,lemma:h_22} it follows that all paths and cycles are either entirely contained in $S$ or in $\overline{S}$, and the paths in $\mathcal{P}$ are not contained in $S$.
    Let $I \subseteq [\![1,p+c]\!]$ be such that $I$ contains all $i$ such that $1 \leq i \leq p$ and $P_i \subseteq \overline{S}$ and all $p+j$ such that $1 \leq j \leq c$ and $C_j \subseteq \overline{S}$.

    It holds that $ \displaystyle \sum_{i \in I} v_i = \sum_{P_i \subseteq \overline{S}} |P_i| + \sum_{C_j \subseteq \overline{S}} |C_j| = |V|-|S| - |\mathcal{P}| = T_V$.

    Also, we know that all vertices in $S^*$ are satisfied with respect to $S$ and hence if $u^* \in S$, the following inequality holds for $u^*$:
    $$ \frac{d(u^*)}{|V|-1} \geq \frac{d_{\overline{S}}(u^*)}{|V|-|S|}. $$
    Recall, that $w_i = (d_{P_i}(u^*), d_{P_i}(v^*))$. Let $w_{i,1}$ be the first element of $w_i$, i.e.\ $w_{i,1} = d_{P_i}(u^*)$. Then, $\displaystyle\sum_{i \in I} w_{i,1} = \sum_{P_i \subseteq \overline{S}} d_{P_i}(u^*) = d_{\overline{S} \setminus \{v^*\}}(u^*)$. Then, $W_1 \geq d_{\overline{S} \setminus \{v^*\}}(u^*)$. A similar argument holds for the vertex $v^*$ if $v^* \in S$. Hence, $\sum_{i \in I} w_i \le W$ and $\mathcal{I}$ is a yes-instance of the \textsc{exact 2d Knapsack Problem}.

    \item[$\impliedby$] Let $I \subseteq [\![1,p+c]\!]$ and let $S = S^* \cup \bigcup_{1\leq i \leq p, i \in I} P_i \cup \bigcup_{1 \leq j \leq c, p+j \in I} C_j$. By the definition of $S$, all vertices in paths and cycles in $S$ are satisfied. Indeed for any given vertex $v$ in $S \cap P_i$ (resp.\ $S \cap C_j$) all neighbors of $v$ are also in $S$. The vertices of $S^*$ are also satisfied with respect to $S$ since $\sum_{i \in I} w_i \leq W$. Finally, $S = |V| - |\mathcal{P}| - \sum_{i \in I} v_i = k$ and hence we conclude that $S$ is indeed a \PDS{} of $G$ of size $k$. This concludes the proof.
\end{itemize}
\end{proof}

Finally, in the theorem below we show that \maxPDS{} can be solved in linear time in graphs with $h\le 2$.

\begin{theorem}
\label{trm:h=2}
\maxPDS{} can be solved in time $\mathcal{O}(n^5)$ in graphs with $h(G)\le 2$.
\end{theorem}

\begin{proof}
Let $G=(V,E)$ be a graph with $h(G)\le 2$. Note that $|E| \leq \mathcal{O}(n)$. If $h(G) \leq 1$ then by \cref{lem:h_1_poly}, we can find a \PDS{} of maximum size in $G$ in time $\mathcal{O}(n)$. If $h(G)=2$,
by using BFS algorithm we can compute the vertices  $u^*,v^*$, paths and cycles of $G[V\setminus\{u^*,v^*\}]$ in time $\mathcal{O}(n^2)$.
Then, there are four possible cases for the set $S^* \subseteq \{u^*,v^*\}$ and $\mathcal{O}(n)$ possible values of $k$. Constructing the equivalent instance of the \textsc{exact 2d Knapsack Problem} can be done in $\mathcal{O}(n)$. Solving this instance can be done in $\mathcal{O}( n \cdot T_V \cdot \min(W_1,n) \cdot \min(W_2,n)) = \mathcal{O}(n^4)$. Thus, the total complexity is $\mathcal{O}(n^5)$.
\end{proof}

Notice, that if there does not exist a \PDS{} of size $k$ with $\frac{\vert V \vert}{2} + 1 \leq k \leq \vert V \vert$, then by \cref{lowerboundS}, we know that the size of a maximum \PDS{} in a given graph equals $\lceil \frac{|V|}{2} \rceil$ if the parity of $|V|$ is even and $\lceil \frac{|V|}{2} \rceil +1$ if the parity of $|V|$ is odd. Moreover, such \PDS{} can be computed in linear time by using the algorithm introduced in \cite{MR4023158}.

Hereafter we explain how to use a similar approach to solve \connmaxPDS{}. Let $G=(V,E)$ be a graph with $|V|=n$. Let $C$ be a connected component of $G$ which forms a cycle, let $v\in V[C]$ and $u,w$ be its neighbors. Let $S$ be a \PDS{} in $G$. Assume that $v\in S$. Then, if $u,w \in \overline{S}$, $|N_S(v)| = 0$ and $v$ is not satisfied with respect to $S$. On the other hand, if $u,w \in S$, then $|N_{\overline{S}}(v)| = 0$ and $v$ is satisfied with respect to $S$. Let us consider the case when $|N_S(v)| = |N_{\overline{S}}(v)|= 1$. Suppose for a contradiction that $|S| > \lceil \frac{n}{2} \rceil \ge \lceil \frac{n}{2} \rceil+1 $ and $v$ is satisfied with respect to $S$. Then, the following inequality holds for the vertex $v$:
\[  \frac{1}{|S|-1}\ge \frac{1}{\lceil \frac{n}{2} \rceil}\ge\frac{1}{|\overline{S}|}\ge \frac{1}{n-(\lceil \frac{n}{2} \rceil+1)}\,,\]
a contradiction. Hence, $|S| \le \lceil\frac{n}{2} \rceil$. Finally, notice that if we take any $\lceil\frac{n}{2} \rceil$ consecutive vertices of a cycle, it forms a \PDS{} of maximum size. A similar argument holds for a connected component $P$ of $G$ which forms a path. This implies the following polynomial-time algorithm for solving \connmaxPDS{}.

Let $G=(V,E)$ be a graph. The case when $h(G)= 1$ is trivial. Assume that $h(G)= 2$. First, we use the BFS algorithm to determine the connected components of a graph. \Cref{lemma:h_21} implies that $G$ is a disjoint union of paths $P_1,\ldots,P_k$, cycles $C_1,\ldots,C_\ell$, for some $k,\ell \ge 0$, and the connected component containing the vertices $u^*,v^*$. Let us denote the connected component containing the vertices $u^*,v^*$ by $C^*$. Let $G^* = G[C^*]$. We use the reduction from \cref{lemma:h=2reduction} to determine the maximum size of a connected \PDS{} in $G^*$. Then, by the argument above, we know that in $C_i$, for $i \in \{1,\ldots,\ell\}$, and in $P_j$, for $j \in \{1,\ldots,k\}$, the maximum size of a connected \PDS{} equals to $\lceil\frac{|C_i|}{2}\rceil$ and $\lceil\frac{|P_j|}{2}\rceil$, respectively. We compare the values and determine the maximum size of a connected \PDS{} in $G$. 

\subsection{Finding \PDS{} in graphs $G$ such that $h(\overline{G}) \leq 2$}
\label{subsec:graphsofdegreeatleastN-3}
In this subsection, we introduce a polynomial-time algorithm for finding a \PDS{} of maximum size in graphs $G$ such that $h(\overline{G})\le 2$. 

First, below we show that, given a graph $G=(V,E)$ on $n$ vertices, $S \subseteq V$ of size at least $\lceil\frac{n}{2}\rceil+1$, and $u \in V$ such that $d_{\overline{G}}(u) \leq 2$, $u$ is satisfied with respect to $S$ if and only if at least one of its neighbors in $\overline{G}$ is not contained in $S$.

Suppose for a contradiction that both neighbors of $u$ in $\overline{G}$, say $v,w$, are contained in $S$ and $u$ is satisfied. Then $\frac{d_S(u)}{|S|-1} \leq \frac{|S|-3}{|S|-1} < 1$ and $\frac{d_{\overline{S}}(u)}{|\overline{S}|}=1$, a contradiction. If $u$ has only one neighbor in $\overline{G}$, say $v$, and $v \in S$, we similarly obtain a contradiction.

Then, we show that if exactly one of the neighbors of $u$ is contained in $S$, then $u$ is satisfied. Without loss of generality, assume that $v \notin S$. Recall that $|S| \geq \lceil\frac{n}{2}\rceil+1$ and $|S|-1 \geq |\overline{S}|$. Thus,
$$ \frac{d_S(u)}{|S|-1} \geq \frac{|S|-2}{|S|-1} = 1 - \frac{1}{|S|-1} \geq 1 - \frac{1}{|\overline{S}|} = \frac{|\overline{S}|-1}{|\overline{S}|} \geq \frac{d_{\overline{S}}(u)}{|\overline{S}|}. $$
Hence, $u$ is satisfied with respect to $S$. Let $G=(V,E)$ be a graph. We say that a subset $D\subseteq V$ dominates $G$ if every vertex of $G$ is either in $D$ or adjacent to a vertex in $D$. The observations above imply the following.

\begin{remark}\label{rk:PDS_to_dom}
    Let $G=(V,E)$ be a graph on $n$ vertices such that $h(\overline{G}) \leq 2$. Let $u^*,v^*$ be two vertices of maximum degree in $\overline{G}$. Let $S \subseteq V$ be a set of vertices of size at least $\lceil\frac{n}{2}\rceil+1$. 
    Then, $S$ forms a \PDS{} in $G$ if and only if the following two conditions are satisfied.
    \begin{itemize}
        \item[$(i)$] The vertices in $\{u^*,v^*\} \cap S$ are satisfied with respect to $S$.

        \item[$(ii)$] Every non-universal vertex in $S$ in $G$ has at least one non-neighbor (i.e. neighbor in $\overline{G}$) outside of $S$. Equivalently, $\overline{S}$ dominates every non-isolated vertex of $\overline{G}$.
    \end{itemize}
\end{remark}

\begin{proof}
By the argument above, we know that vertices $u$ with $d_{\overline{G}}(u) \leq 2$ are satisfied with respect to $S$ if and only if one of the two neighbors of $u$ in $\overline{G}$ is contained in $\overline{S}$. Since $S$ forms a \PDS{} in $G$ if and only if all the vertices in $S$ are satisfied and $u^*,v^*$ are satisfied only if they have at least one of their neighbors in $\overline{G}$ in $\overline{S}$, the statement follows.
\end{proof}
 
Further, we show that finding a \PDS{} of maximum size in $G$ such that $h(\overline{G})\le 2$ can be done in polynomial time. Similarly as before, we first consider the case when $h(\overline{G})=1$.

\begin{lemma}
\label{lemma:h=1comp}
    \maxPDS{} can be solved in polynomial time on graphs $G$ such that $h(\overline{G})=1$.
\end{lemma}

\begin{proof}
Recall that $V = V_I \cup V_M \cup V_*$where $V_I, V_M$ and $V_*$ represent an independent set, matching and a star in $\overline{G}$, respectively. Let $u^*$ be the center of the star $V^*$, i.e. the vertex of maximum degree in $\overline{G}$. Notice that $|V_M| = 2k$, for some positive integer $k$. Let $V_M = \{u_1,v_1,...,u_k,v_k\}$, where $u_iv_i \in E(\overline{G})$, for all $i \in [\![1,k]\!]$.

Further, we show that $S = V_I \cup \{u_1,...,u_k\} \cup (V_* \setminus \{u^*\})$ forms a \PDS{} of maximum size in $G$. Notice that $|S| \geq \lceil \frac{n}{2} \rceil +1$ and
\begin{itemize}
    \item[$(i)$] $S$ only contains vertices of degree at most $2$ in $\overline{G}$, i.e. it does not contain $u^*$;
    \item[$(ii)$] every non universal vertex of $S$ has a non-neighbor in $\overline{S}$.
\end{itemize}
For every \PDS{} $S'$ in $G$, to ensure that all vertices in $S'$ are satisfied, we must include at least one of $u_i$ or $v_i$, for every $i\in \{1,\ldots,k\}$, in $\overline{S'}$. Similarly, we must include at least one vertex of $V^*$ in $\overline{S'}$. Thus, $|S'| \leq |S|$, for all \PDS{}'s $S'$ in $G$. This implies that $S$ is a \PDS{} of maximum size in $G$. 
\end{proof}

Note that the \PDS{} considered in the proof of the \cref{lemma:h=1comp} induces a connected subgraph of $G$. Indeed, if $\overline{G}$ contains at least two connected components, then $\overline{G}[S]$ also contains at least two connected components. In this case, $G[S]$ is clearly connected. Otherwise, either $V_I=\emptyset$ and $V_M = \emptyset$ which implies $S = V_* \setminus \{u^*\}$ is an independent set in $\overline{G}$ and a clique in $G$, or $V_I = \emptyset$, $V_* = \emptyset$ and $V_M = \{u_1,v_1\}$ which implies that $|V|=2$, a contradiction to the fact that $|V| \geq 3$.

Let us now consider the case when  $h(\overline{G})=2$. We will use the structure of $\overline{G}$ (see \cref{lemma:h_21}) and construct a dominating set $D$ in $\overline{G}$ such that $S = V \setminus D$ forms a \PDS{} in $G$ (see \cref{rk:PDS_to_dom}).

\begin{theorem}
\label{thrm:h=2comp}
\maxPDS{} can be solved in time $\mathcal{O}(n^*)$ in graphs $G$ such that $h(\overline{G})=2$.
\end{theorem}

\begin{proof}
\Cref{lemma:h_21} implies that $\overline{G}$ has the following structure. Let $u^*,v^*$ be the two vertices of $\overline{G}$ of maximum degree. The connected components of $\overline{G}[V\setminus \{u^*v^*\}]$ can be categorized into the following types:
\begin{itemize}
    \item Type 1. Cycles or paths with adjacencies to neither $u^*$ nor $v^*$;
    \item Type 2. Paths with one endpoint adjacent to $u^*$;
    \item Type 3. Paths with both endpoints adjacent to $u^*$;
    \item Type 4. Paths with one endpoint adjacent to $v^*$;
    \item Type 5. Paths with one endpoint adjacent to $v^*$;
    \item Type 6. Paths with one endpoint adjacent to $u^*$ and one endpoint adjacent to $v^*$.
\end{itemize}

Below, we consider different subtypes of the types of the components of $\overline{G}[V \setminus \{u^*,v^*\}]$ depending on the cardinality of the connected component modulo $3$. More specifically, we consider three cases. Case when the cardinality of the connected component equals $0 \mod 3$, $1 \mod 3$ and finally, $2 \mod 3$. We denote these subtypes as follows. For $x\in \{1,\ldots,6\}$, types $x.1$, $x.2$ and $x.3$ refers to the connected components of type $x$ containing $3k$, $3k+1$ and $3k+2$ vertices for some positive integer $k$, respectively. 

\paragraph{Procedure:} As mentioned earlier, we will construct sets $D$ in $\overline{G}$ and check whether $S = V \setminus D$ forms a \PDS{} in $G$. Recall that \Cref{rk:PDS_to_dom} implies that, for every \PDS{} $S$ in $G$, $D=\overline{S}$ must dominate every non-isolated vertex in $\overline{G}$. 


More specifically, we proceed as follows. We construct various sets $D$ with the smallest size possible while ensuring they meet the required properties. Subsequently, we check if each corresponding set $S = \overline{D}$ forms a \PDS{}. If none of the sets of such size satisfy this condition, we conclude that there does not exist a  \PDS{} of size $n-|D|$. We then proceed to construct sets $D$ of larger sizes. Notice that if there exists a set $D$ containing either $u^*$ or/and $v^*$ and another set $D'$ of the same size that doesn't contain those vertices, we favor the set $D$. This comes from the fact that $u^*$ and $v^*$ have the lowest degree in $G$. If included in $S$, we must maximize the inclusion of their neighbors in $\overline{G}$ to ensure that they are satisfied with respect to $S$.


We construct $D$ of the smallest possible size for each type and subtype presented above as follows. Let $P$ be a connected component of $\overline{G}[V \setminus \{u^*,v^*\}]$ such that $|P|=\ell$. First, we consider connected components of Type 1.


\begin{itemize}
    \item Type 1: Let us denote the vertices in $P$ by $v_1,\ldots,v_{\ell}$ in such a way that, for every $i \in \{1,...,\ell\}$, $v_{i}$ is adjacent in $\overline{G}$ to $v_{i-1}$ and $v_{i+1}$, if such vertices exist. \Cref{rk:PDS_to_dom} implies that to ensure that the vertices in $S=\overline{D}$ are satisfied, $D$ contains at least one vertex from $\{v_{i-1},v_i,v_{i+1}\}$, for every $i \in \{2,...,\ell-1\}$. Then, since $D$ is of a minimum size, it contains $\lceil\frac{|P|}{3}\rceil$ vertices in $P$.
\end{itemize}

Let us consider the remaining cases. Let $ u $ be a neighbor of $ u^* $ and $ u' $ be a neighbor of $ u $ in $ P $. Similarly, let $ v $ be a neighbor of $ v^* $ and $ v' $ be a neighbor of $ v $ in $ P $. Regardless of the cardinality of $ P $ modulo 3, a minimum size set $ D $ contains at least $\left\lceil \frac{|P|}{3} \right\rceil - 1$ vertices from $ P $ that are not $ u, u', v, v', u^*, $ or $ v^* $. Then, depending on the cardinality of $ P $ modulo 3, additional neighbors must be included in $ D $ to ensure all vertices are satisfied.

\begin{itemize}
    \item Type 2. Below, we consider different subtypes of the components of $\overline{G}[V \setminus \{u^*,v^*\}]$ of type $2$ depending on the cardinality of the connected component modulo $3$.
    \begin{enumerate}
        \item[2.0] In this case we include $u'$ in $D$.
        \item[2.1] In this case we include $u^*$ in $D$.
        \item[2.2] In this case we include $u$ in $D$.
    \end{enumerate}

    \item Type 3. Below, we consider different subtypes of the components of $\overline{G}[V \setminus \{u^*,v^*\}]$ of type $3$ depending on the cardinality of the connected component modulo $3$.
    \begin{enumerate}
        \item[3.0] In this case we include $u'$ in $D$.
        \item[3.1] In this case we include $u^*$ in $D$.
        \item[3.2] In this case we include $u^*$ in $D$.
    \end{enumerate}

\item Type 4: This case is analogous to Type $2$, but with $ u^* $ replaced by $ v^* $.

\item Type 5: This case is analogous to Type $3$, but with $ u^* $ replaced by $ v^* $.
\end{itemize}

If, after considering the cases above, either $ u^* $ or $ v^* $ is contained in $ D $, then we can remove these vertices from the graph $\overline{G}$ as well as their neighbors in $\overline{G}$ (except $ u^* $ or $ v^* $ if they are adjacent in $\overline{G}$) when considering components of Type $6$. 
Indeed, let $ v $ be such that $ d_{\overline{G}}(v) \leq 2 $ and $ v $ is adjacent to $ u^* $ in $\overline{G}$, and let $ v' $ be its other neighbor. 
If $ u^* \in D $ and we want $ D $ to be of minimum size, it is better to add $ v' $ to $ D $ rather than $ v $. 
Thus, in this case, we can assume that $ v \notin D $. 
A similar argument applies if $ v^* \in D $. Moreover, if either $ u^* \in D $ or $ v^* \in D $, and we follow the steps outlined above, the Type $6$ components will transform into components of Type $1, 2,$ or $4$, and we proceed accordingly. 
If neither $ u^* \in D $ nor $ v^* \in D $, we proceed as follows.

Similarly as before, let $ u $ be a neighbor of $ u^* $ and $ u' $ be a neighbor of $ u $ in $ P $. Similarly, let $ v $ be a neighbor of $ v^* $ and $ v' $ be a neighbor of $ v $ in $ P $. Regardless of the cardinality of $ P $ modulo 3, a minimum size set $ D $ contains at least $\left\lceil \frac{|P|}{3} \right\rceil - 1$ vertices from $ P $ that are not $ u, u', v, v'$ and either not $ u^*$ or $ v^* $. Then, depending on the cardinality of $ P $ modulo 3, additional vertices must be included in $ D $ to ensure all vertices are satisfied. 

\begin{itemize}
\item Type $6$. Below, we consider different subtypes of the components of $\overline{G}[V \setminus \{u^*,v^*\}]$ of type $6$ depending on the cardinality of the connected component modulo $3$.
\begin{enumerate}
    \item[6.0] In this case, if $ u' \in D $, we include $ v' $ in $ D $, and if $ v' \in D $, we include $ u' $ in $ D $.
    \item[6.1] In this case we include either $u^*$ or $v^*$ in $D$.
    \item[6.2] In this case, if $u' \in D$, we include $v$, and if $v'\in D$, we include $u$ in $D$.
\end{enumerate}
\end{itemize}

On \cref{fig:dom_each_type}, we illustrate the connected components of Types $2$, $3$, and $6$. Note that components of Type $1$ are trivial. Components of Types $4$ and $5$ are analogous to Types $2$ and $3$, respectively, but with $u^*$ replaced by $v^*$. The red vertices represent those included in the set $D$ of minimum size.

\begin{figure}[!ht]
\centering
\begin{tikzpicture}

\node () at (-1,1) {\Large \underline{Type 2:}};
\node () at (-4,-0.1) {type 2.0};
\node () at (-4,-1.1) {type 2.1};
\node () at (-4,-2.1) {type 2.2};
\node () at (0,0) {
\begin{tikzpicture}[scale=0.5]
\draw[thick,black] (0.00,0.00) -- (-1.00,0.00);
\draw[thick,black] (-1.00,0.00) -- (-2.00,0.00);
\draw[thick,black] (-2.00,0.00) -- (-3.00,0.00);
\draw[thick,black,dashed] (-3.00,0.00) -- (-7.00,0.00);
\draw[thick,black] (-7.00,0.00) -- (-8.00,0.00);
\draw[thick,black] (-8.00,0.00) -- (-9.00,0.00);

\filldraw[blue] (0.00,0.00) circle (3pt);
\node () at (0.7,0.3) {\large $u^*$};
\filldraw[blue] (-1.00,0.00) circle (3pt);
\filldraw[red] (-2.00,0.00) circle (3pt);
\filldraw[blue] (-3.00,0.00) circle (3pt);
\filldraw[blue] (-7.00,0.00) circle (3pt);
\filldraw[red] (-8.00,0.00) circle (3pt);
\filldraw[blue] (-9.00,0.00) circle (3pt);


\draw[thick,red] (-2,0) ellipse (0.3cm and 0.3cm);
\draw[thick,red] (-8,0) ellipse (0.3cm and 0.3cm);
\end{tikzpicture}
};
\node () at (0,-1) {
\begin{tikzpicture}[scale=0.5]
\draw[thick,black] (0.00,0.00) -- (1.00,0.00);
\draw[thick,black] (0.00,0.00) -- (-1.00,0.00);
\draw[thick,black] (-1.00,0.00) -- (-2.00,0.00);
\draw[thick,black] (-2.00,0.00) -- (-3.00,0.00);
\draw[thick,black,dashed] (-3.00,0.00) -- (-6.00,0.00);
\draw[thick,black] (-6.00,0.00) -- (-7.00,0.00);
\draw[thick,black] (-7.00,0.00) -- (-8.00,0.00);

\filldraw[red] (1.00,0.00) circle (3pt);
\node () at (1.7,0.3) {\large $u^*$};
\filldraw[blue] (0.00,0.00) circle (3pt);
\filldraw[blue] (-1.00,0.00) circle (3pt);
\filldraw[red] (-2.00,0.00) circle (3pt);
\filldraw[blue] (-3.00,0.00) circle (3pt);
\filldraw[blue] (-6.00,0.00) circle (3pt);
\filldraw[red] (-7.00,0.00) circle (3pt);
\filldraw[blue] (-8.00,0.00) circle (3pt);


\draw[thick,red] (1,0) ellipse (0.3cm and 0.3cm);
\draw[thick,red] (-2,0) ellipse (0.3cm and 0.3cm);
\draw[thick,red] (-7,0) ellipse (0.3cm and 0.3cm);
\end{tikzpicture}};
\node () at (0,-2) {\begin{tikzpicture}[scale=0.5]
\draw[thick,black] (2.00,0.00) -- (1.00,0.00);
\draw[thick,black] (0.00,0.00) -- (1.00,0.00);
\draw[thick,black] (0.00,0.00) -- (-1.00,0.00);
\draw[thick,black] (-1.00,0.00) -- (-2.00,0.00);
\draw[thick,black] (-2.00,0.00) -- (-3.00,0.00);
\draw[thick,black,dashed] (-3.00,0.00) -- (-5.00,0.00);
\draw[thick,black] (-5.00,0.00) -- (-6.00,0.00);
\draw[thick,black] (-7.00,0.00) -- (-6.00,0.00);

\filldraw[blue] (2.00,0.00) circle (3pt);
\node () at (2.7,0.3) {\large $u^*$};
\filldraw[red] (1.00,0.00) circle (3pt);
\filldraw[blue] (0.00,0.00) circle (3pt);
\filldraw[blue] (-1.00,0.00) circle (3pt);
\filldraw[red] (-2.00,0.00) circle (3pt);
\filldraw[blue] (-3.00,0.00) circle (3pt);
\filldraw[blue] (-5.00,0.00) circle (3pt);
\filldraw[red] (-6.00,0.00) circle (3pt);
\filldraw[blue] (-7.00,0.00) circle (3pt);


\draw[thick,red] (1,0) ellipse (0.3cm and 0.3cm);
\draw[thick,red] (-2,0) ellipse (0.3cm and 0.3cm);
\draw[thick,red] (-6,0) ellipse (0.3cm and 0.3cm);
\end{tikzpicture}};

\node () at (8,1) {\Large \underline{Type 6:}};
\node () at (11,-0.1) {type 6.0};
\node () at (11,-1.1) {type 6.1};
\node () at (11,-2.1) {type 6.2};
\node () at (7,0) {\begin{tikzpicture}[scale=0.5]
\draw[thick,black] (0.00,0.00) -- (-1.00,0.00);
\draw[thick,black] (-1.00,0.00) -- (-2.00,0.00);
\draw[thick,black] (-2.00,0.00) -- (-3.00,0.00);
\draw[thick,black,dashed] (-3.00,0.00) -- (-7.00,0.00);
\draw[thick,black] (-7.00,0.00) -- (-8.00,0.00);
\draw[thick,black] (-8.00,0.00) -- (-9.00,0.00);
\draw[thick,black] (-9.00,0.00) -- (-10.00,0.00);

\filldraw[blue] (0.00,0.00) circle (3pt);
\node () at (0.7,0.3) {\large $u^*$};
\filldraw[blue] (-10.00,0.00) circle (3pt);
\node () at (-10.7,0.3) {\large $v^*$};
\filldraw[blue] (-1.00,0.00) circle (3pt);
\filldraw[red] (-2.00,0.00) circle (3pt);
\filldraw[blue] (-3.00,0.00) circle (3pt);
\filldraw[blue] (-7.00,0.00) circle (3pt);
\filldraw[red] (-8.00,0.00) circle (3pt);
\filldraw[blue] (-9.00,0.00) circle (3pt);


\draw[thick,red] (-2,0) ellipse (0.3cm and 0.3cm);
\draw[thick,red] (-8,0) ellipse (0.3cm and 0.3cm);
\end{tikzpicture}};
\node () at (7,-1) {\begin{tikzpicture}[scale=0.5]
\draw[thick,black] (0.00,0.00) -- (1.00,0.00);
\draw[thick,black] (0.00,0.00) -- (-1.00,0.00);
\draw[thick,black] (-1.00,0.00) -- (-2.00,0.00);
\draw[thick,black] (-2.00,0.00) -- (-3.00,0.00);
\draw[thick,black,dashed] (-3.00,0.00) -- (-6.00,0.00);
\draw[thick,black] (-6.00,0.00) -- (-7.00,0.00);
\draw[thick,black] (-7.00,0.00) -- (-8.00,0.00);
\draw[thick,black] (-9.00,0.00) -- (-8.00,0.00);

\filldraw[blue] (-9.00,0.00) circle (3pt);
\node () at (-9.7,0.3) {\large $v^*$};

\filldraw[red] (1.00,0.00) circle (3pt);
\node () at (1.7,0.3) {\large $u^*$};
\filldraw[blue] (0.00,0.00) circle (3pt);
\filldraw[blue] (-1.00,0.00) circle (3pt);
\filldraw[red] (-2.00,0.00) circle (3pt);
\filldraw[blue] (-3.00,0.00) circle (3pt);
\filldraw[blue] (-6.00,0.00) circle (3pt);
\filldraw[red] (-7.00,0.00) circle (3pt);
\filldraw[blue] (-8.00,0.00) circle (3pt);


\draw[thick,red] (1,0) ellipse (0.3cm and 0.3cm);
\draw[thick,red] (-2,0) ellipse (0.3cm and 0.3cm);
\draw[thick,red] (-7,0) ellipse (0.3cm and 0.3cm);
\end{tikzpicture}};
\node () at (7,-2) {\begin{tikzpicture}[scale=0.5]
\draw[thick,black] (2.00,0.00) -- (1.00,0.00);
\draw[thick,black] (0.00,0.00) -- (1.00,0.00);
\draw[thick,black] (0.00,0.00) -- (-1.00,0.00);
\draw[thick,black] (-1.00,0.00) -- (-2.00,0.00);
\draw[thick,black] (-2.00,0.00) -- (-3.00,0.00);
\draw[thick,black,dashed] (-3.00,0.00) -- (-5.00,0.00);
\draw[thick,black] (-5.00,0.00) -- (-6.00,0.00);
\draw[thick,black] (-7.00,0.00) -- (-6.00,0.00);
\draw[thick,black] (-7.00,0.00) -- (-8.00,0.00);

\filldraw[blue] (-8.00,0.00) circle (3pt);
\node () at (-8.7,0.3) {\large $v^*$};
\filldraw[blue] (2.00,0.00) circle (3pt);
\node () at (2.7,0.3) {\large $u^*$};
\filldraw[red] (1.00,0.00) circle (3pt);
\filldraw[blue] (0.00,0.00) circle (3pt);
\filldraw[blue] (-1.00,0.00) circle (3pt);
\filldraw[red] (-2.00,0.00) circle (3pt);
\filldraw[blue] (-3.00,0.00) circle (3pt);
\filldraw[blue] (-5.00,0.00) circle (3pt);
\filldraw[red] (-6.00,0.00) circle (3pt);
\filldraw[blue] (-7.00,0.00) circle (3pt);


\draw[thick,red] (1,0) ellipse (0.3cm and 0.3cm);
\draw[thick,red] (-2,0) ellipse (0.3cm and 0.3cm);
\draw[thick,red] (-6,0) ellipse (0.3cm and 0.3cm);
\end{tikzpicture}};

\node () at (3.5,-3.5) {\Large \underline{Type 3:}};
\node () at (0,-5) {type 3.0};
\node () at (7,-5) {type 3.1};
\node () at (3.5,-7.5) {type 3.2};
\node () at (0,-5) {\begin{tikzpicture}[scale=0.5]
\draw[thick] (-7.50,1.73) arc [start angle=120.00, end angle=-120.00, x radius=5, y radius=2];
\draw[thick,dashed] (-7.50,1.73) arc [start angle=120.00, end angle=240.00, x radius=5, y radius=2];

\filldraw[blue] (0.00,0.00) circle (3pt);
\filldraw[blue] (-2.50,1.73) circle (3pt);
\filldraw[red] (-5.00,2.00) circle (3pt);
\filldraw[blue] (-7.50,1.73) circle (3pt);
\filldraw[blue] (-7.50,-1.73) circle (3pt);
\filldraw[red] (-5.00,-2.00) circle (3pt);
\filldraw[blue] (-2.50,-1.73) circle (3pt);
\draw[thick,red] (-5.00,2.00) ellipse (0.3cm and 0.3cm);
\draw[thick,red] (-5.00,-2.00) ellipse (0.3cm and 0.3cm);

\node () at (0.7,0.3) {\large $u^*$};
\end{tikzpicture}};
\node () at (7,-5) {\begin{tikzpicture}[scale=0.5]
\draw[thick] (-7.50,1.73) arc [start angle=120.00, end angle=-120.00, x radius=5, y radius=2];
\draw[thick,dashed] (-7.50,1.73) arc [start angle=120.00, end angle=240.00, x radius=5, y radius=2];

\filldraw[red] (0.00,0.00) circle (3pt);
\filldraw[blue] (-2.50,1.73) circle (3pt);
\filldraw[red] (-5.00,2.00) circle (3pt);
\filldraw[blue] (-7.50,1.73) circle (3pt);
\filldraw[blue] (-7.50,-1.73) circle (3pt);
\filldraw[red] (-5.00,-2.00) circle (3pt);
\filldraw[blue] (-2.50,-1.73) circle (3pt);
\filldraw[blue] (-0.67,-1.00) circle (3pt);
\draw[thick,red] (0.00,0.00) ellipse (0.3cm and 0.3cm);
\draw[thick,red] (-5.00,2.00) ellipse (0.3cm and 0.3cm);
\draw[thick,red] (-5.00,-2.00) ellipse (0.3cm and 0.3cm);

\node () at (0.7,0.3) {\large $u^*$};
\end{tikzpicture}};
\node () at (3.5,-7.5) {\begin{tikzpicture}[scale=0.5]
\draw[thick] (-7.50,1.73) arc [start angle=120.00, end angle=-120.00, x radius=5, y radius=2];
\draw[thick,dashed] (-7.50,1.73) arc [start angle=120.00, end angle=240.00, x radius=5, y radius=2];

\filldraw[red] (0.00,0.00) circle (3pt);
\filldraw[blue] (-0.67,1.00) circle (3pt);
\filldraw[blue] (-2.50,1.73) circle (3pt);
\filldraw[red] (-5.00,2.00) circle (3pt);
\filldraw[blue] (-7.50,1.73) circle (3pt);
\filldraw[blue] (-7.50,-1.73) circle (3pt);
\filldraw[red] (-5.00,-2.00) circle (3pt);
\filldraw[blue] (-2.50,-1.73) circle (3pt);
\filldraw[blue] (-0.67,-1.00) circle (3pt);
\draw[thick,red] (0.00,0.00) ellipse (0.3cm and 0.3cm);
\draw[thick,red] (-5.00,2.00) ellipse (0.3cm and 0.3cm);
\draw[thick,red] (-5.00,-2.00) ellipse (0.3cm and 0.3cm);

\node () at (0.7,0.3) {\large $u^*$};
\end{tikzpicture}
};

\end{tikzpicture}

\caption{In this figure, vertices in red represent the vertices included in the minimum size set $D$. For components of types $6.1$ and $6.2$, we depict only one of the two possible cases; the other cases, where $v^* \in D$ or $v^*$ has a neighbor in $D$, are analogous.}
\label{fig:dom_each_type}
\end{figure}
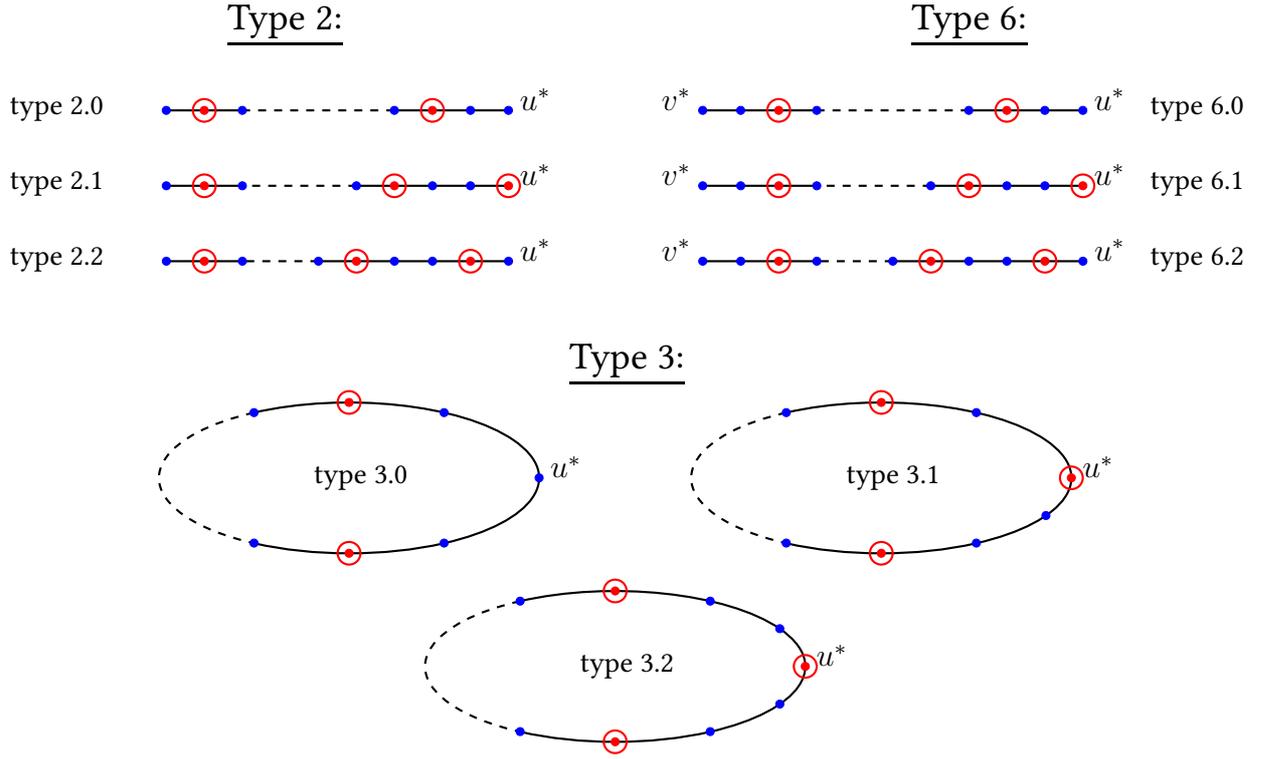

If there are connected components of Type $6.1$, we consider them first among the components of type $6$. We construct two sets $D$, one with $u^*\in D$ and another with $v^* \in D$. Similarly as before, we remove the vertex included in $D$ along with its neighbors from $\overline{G}$. Then, the components of Type $6$ will transform into components of Type $1,2$ or $4$.
Then, if components of Type $6.2$ exist, we denote by $k$ be the number of such components, then, for all $i \in [\![0,k]\!]$, we create a set $D$ containing $i$ vertices adjacent to $u^*$ and $k-i$ vertices adjacent to $v^*$.

For each set $D$ constructed through the procedure introduced above, we check whether $S = V \setminus D$ forms a PDS in $G$. 
If none of these sets forms a PDS in $G$, we iterate the procedure with two cases: one with $u^* \in D$ and the other with $v^* \in D$. Similarly, if none of the constructed sets $D$ form a \PDS{},  we repeat the procedure with $u^*,v^* \in D$. Note, that if $u^*,v^* \in D$, $S = V\setminus D$ forms  a \PDS{} in $G$. Indeed, $u^*,v^* \notin S$ and $D=\overline{S}$ dominates all non-isolated vertices in $\overline{G}$ and \cref{rk:PDS_to_dom} implies that such $S$ forms a \PDS{}.

\paragraph{Correctness:} As outlined earlier and implied by \cref{rk:PDS_to_dom}, the procedure introduced above always outputs a set $D$ such that $S = V \setminus D$ forms a \PDS{} in $G$. This is achieved by ensuring that all non-isolated vertices of $\overline{G}$ are dominated, and that the vertices in $\{u^*, v^*\} \cap S$ are satisfied with respect to $S$ in $G$. Furthermore, the set $D$ is minimal, thereby justifying the correctness of the algorithm.

\paragraph{Complexity:} By using BFS algorithm, we can determine the connected components of $\overline{G}[V \setminus \{u^*,v^*\}]$ in time $\mathcal{O}(n \cdot |E(\overline{G})|) = \mathcal{O}(n^2)$. Then, constructing sets $D$ is done in time $\mathcal{O}(n)$. Thus, the overall  complexity is $\mathcal{O}(n^2)$.
\end{proof}

Note that the result introduced above implies a polynomial-time complexity on graphs $G$ such that $\Delta(\overline{G}) \leq 2$, i.e. graphs $G=(V,E)$ such that $\delta(G) \geq |V|-3$.

Note that the \PDS{} constructed in the proof of the \cref{thrm:h=2comp} induces a connected subgraph of $G$.
Indeed, if $u^*$ and $v^*$ are universal in $\overline{G}$, then $V \setminus \{u^*,v^*\}$ forms an independent set in $\overline{G}$ and a clique in $G$.
Otherwise, if $\overline{G}[V \setminus\{u^*,v^*\}]$ contains two connected components or a path with at least 4 vertices, the \PDS{} constructed by the procedure introduced in the proof of \cref{thrm:h=2comp} includes at least two vertices which are not adjacent and have no neighbors in common in $\overline{G}$. Consequently, $S$ always induces a connected subgraph of $G$.

Therefore, we introduced all the results outlined in \cref{table:results}.

\section{Conclusion}
In this paper, we investigated \maxPDS{} and \connmaxPDS{}. We gave new results on the parameterized complexity. The parameters considered in this paper are $\degen, \Delta, h-$index. There remain several open questions, some of which we present hereafter.
\begin{itemize}
    \item As mentioned in \cref{table:results}, the complexity of finding a (connected) \PDS{} of maximum size in graphs with $\Delta(G) =3$,  $h(G)=3$, $3\leq \Delta(\overline{G}) \leq 5$ or $3 \leq h(\overline{G}) \leq 5$ remain open problems.

    \item The complexity of \maxPDS{} and \connmaxPDS{} with respect to other graph parameters, both related and unrelated to graph density could lead to interesting results. 

    \item As mentioned in \cref{sec:introduction}, we know that the size of \PDS{} and connected \PDS{} can differ by a factor of 2 in some graphs. However, \maxPDS{} and \connmaxPDS{} behave similarly with respect to the parameters studied in this paper.
    Therefore, identifying a parameter or a graph class where there is a fundamental complexity difference between the \maxPDS{} and \connmaxPDS{} problems could be an interesting direction for future research.
\end{itemize}
\bibliographystyle{elsarticle-num} 
\bibliography{sn-bibliography}

\begin{thebibliography}{10}
\expandafter\ifx\csname url\endcsname\relax
  \def\url#1{\texttt{#1}}\fi
\expandafter\ifx\csname urlprefix\endcsname\relax\def\urlprefix{URL }\fi
\expandafter\ifx\csname href\endcsname\relax
  \def\href#1#2{#2} \def\path#1{#1}\fi

\bibitem{MR4023158}
C.~Bazgan, J.~Chleb\'{\i}kov\'{a}, C.~Dallard, T.~Pontoizeau, Proportionally dense subgraph of maximum size: complexity and approximation, Discrete Applied Mathematics. The Journal of Combinatorial Algorithms, Informatics and Computational Sciences 270 (2019) 25--36.
\newblock \href {https://doi.org/10.1016/j.dam.2019.07.010} {\path{doi:10.1016/j.dam.2019.07.010}}.

\bibitem{CommunityProtein}
Y.~Zhu, Y.~Li, J.~Liu, L.~Qin, J.~X. Yu, Discovering large conserved functional components in global network alignment by graph matching., BMC Genomics 19 (Suppl 7) 670 (2018).

\bibitem{olsenGenView}
M.~Olsen, A general view on computing communities, Mathematical Social Sciences 66~(3) (2013) 331--336.
\newblock \href {https://doi.org/https://doi.org/10.1016/j.mathsocsci.2013.07.002} {\path{doi:https://doi.org/10.1016/j.mathsocsci.2013.07.002}}.

\bibitem{densestksub}
U.~Feige, D.~Peleg, G.~Kortsarz, The dense k -subgraph problem, Algorithmica 29 (2001) 410–421.
\newblock \href {https://doi.org/https://doi.org/10.1007/s004530010050} {\path{doi:https://doi.org/10.1007/s004530010050}}.

\bibitem{ListOfSimilarProblems}
C.~Komusiewicz, Multivariate algorithmics for finding cohesive subnetworks, Algorithms 9~(1) (2016).
\newblock \href {https://doi.org/10.3390/a9010021} {\path{doi:10.3390/a9010021}}.

\bibitem{bazgan:k-comm}
C.~Bazgan, J.~Chlebikova, T.~Pontoizeau, Structural and algorithmic properties of 2-community structures, Algorithmica 80~(6) (2018) 1890--1908.
\newblock \href {https://doi.org/https://doi.org/10.1007/s00453-017-0283-7} {\path{doi:https://doi.org/10.1007/s00453-017-0283-7}}.

\bibitem{Estivill1}
V.~Estivill-Castro, M.~Parsa, On connected two communities, in: Proceedings of the Thirty-Sixth Australasian Computer Science Conference - Volume 135, ACSC '13, Australian Computer Society, Inc., AUS, 2013, pp. 23--30.

\bibitem{bazgan:firstInfiniteFamily}
C.~Bazgan, J.~Chleb{\'\i}kov{\'a}, C.~Dallard, Graphs without a partition into two proportionally dense subgraphs, Information Processing Letters 155 (2020) 105877.
\newblock \href {https://doi.org/https://doi.org/10.1016/j.ipl.2019.105877} {\path{doi:https://doi.org/10.1016/j.ipl.2019.105877}}.

\bibitem{BaghirovaK-Community}
N.~Baghirova, C.~Dallard, B.~Ries, D.~Schindl, Finding k-community structures in special graph classes (04 2023).
\newblock \href {https://doi.org/10.21203/rs.3.rs-2826086/v1} {\path{doi:10.21203/rs.3.rs-2826086/v1}}.

\bibitem{SatPartition}
C.~Bazgan, Z.~Tuza, D.~Vanderpooten, The satisfactory partition problem, Discrete Applied Mathematics 154~(8) (2006) 1236--1245.
\newblock \href {https://doi.org/https://doi.org/10.1016/j.dam.2005.10.014} {\path{doi:https://doi.org/10.1016/j.dam.2005.10.014}}.

\bibitem{NoteOnSatPart}
F.~Ciccarelli, M.~{Di Ianni}, G.~Palumbo, A note on the satisfactory partition problem: Constant size requirement, Information Processing Letters 179 (2023).
\newblock \href {https://doi.org/https://doi.org/10.1016/j.ipl.2022.106292} {\path{doi:https://doi.org/10.1016/j.ipl.2022.106292}}.

\bibitem{BalSatPart}
A.~Gaikwad, S.~Maity, S.~K. Tripathi, The Balanced Satisfactory Partition Problem, Springer International Publishing, 2021, pp. 322--336.
\newblock \href {https://doi.org/https://doi.org/10.1007/978-3-030-67731-2_23} {\path{doi:https://doi.org/10.1007/978-3-030-67731-2_23}}.

\bibitem{SatPartParamComp}
A.~Gaikwad, S.~Maity, S.~K. Tripathi, Parameterized complexity of satisfactory partition problem, Theoretical Computer Science 907 (2022) 113--127.
\newblock \href {https://doi.org/https://doi.org/10.1016/j.tcs.2022.01.022} {\path{doi:https://doi.org/10.1016/j.tcs.2022.01.022}}.

\bibitem{SatPartAlgApproach}
M.~U. Gerber, D.~Kobler, Algorithmic approach to the satisfactory graph partitioning problem, European Journal of Operational Research 125~(2) (2000) 283--291.
\newblock \href {https://doi.org/https://doi.org/10.1016/S0377-2217(99)00459-2} {\path{doi:https://doi.org/10.1016/S0377-2217(99)00459-2}}.

\bibitem{LCP}
N.~Baghirova, C.~L. Gonzalez, B.~Ries, D.~Schindl, Locally checkable problems parameterized by clique-width, in: 33rd International Symposium on Algorithms and Computation (ISAAC 2022), Vol. 248 of Leibniz International Proceedings in Informatics (LIPIcs), 2022, pp. 31:1--31:20.

\bibitem{PDS}
C.~Bazgan, J.~Chlebikova, C.~Dallard, T.~Pontoizeau, Proportionally dense subgraph of maximum size: Complexity and approximation, Discrete Applied Mathematics 270 (07 2019).
\newblock \href {https://doi.org/10.1016/j.dam.2019.07.010} {\path{doi:10.1016/j.dam.2019.07.010}}.

\bibitem{pontoizeau}
T.~Pontoizeau, Community detection: computational complexity and approximation, Ph.D. thesis, Université Paris sciences et lettres (2018).

\bibitem{indep_cubic}
H.~Fleischner, G.~Sabidussi, V.~I. Sarvanov, Maximum independent sets in 3- and 4-regular hamiltonian graphs, Discrete Mathematics 310~(20) (2010) 2742--2749, graph Theory — Dedicated to Carsten Thomassen on his 60th Birthday.
\newblock \href {https://doi.org/https://doi.org/10.1016/j.disc.2010.05.028} {\path{doi:https://doi.org/10.1016/j.disc.2010.05.028}}.

\bibitem{garey_johnson}
M.~R. Garey, D.~S. Johnson, Computers and Intractability; A Guide to the Theory of NP-Completeness, W. H. Freeman \& Co., USA, 1990.

\end{thebibliography}

\end{document}